\documentclass[twocolumn,twoside]{article}
\usepackage{amsmath, amssymb}

\usepackage[T1]{fontenc}
\usepackage{tikz}
\usepackage{tikz-3dplot}
\usepackage{xcolor}
\usepackage{cite}
\usetikzlibrary{trees,positioning,decorations.pathreplacing,calc,patterns,patterns.meta,fit,perspective,angles,quotes}
\usepackage{color}
\usepackage{graphicx}
\usepackage{calc}
\usepackage{amsfonts}
\usepackage{xstring}
\usepackage[pdfencoding=auto]{hyperref}
\usepackage{cleveref}
\usepackage{caption}
\usepackage{subcaption}
\usepackage{ifthen}
\usepackage[misc]{ifsym}
\usepackage{pifont}
\usepackage{mathtools}
\usepackage{trimclip}
\usepackage{multirow}

\usepackage{algorithm}
\usepackage{algpseudocodex}
\algnewcommand{\LineComment}[1]{\State \(\triangleright\) #1}

\newtheorem{theorem}{Theorem}
\newtheorem{fact}[theorem]{Fact}
\newtheorem{observation}[theorem]{Observation}
\newtheorem{lemma}{Lemma}
\newtheorem{corollary}{Corollary}

\renewcommand{\subsection}[1]{\textbf{#1}.}
\renewcommand{\subsubsection}[1]{\textbf{#1}.}

\definecolor{okabe1}{HTML}{000000}
\definecolor{okabe2}{HTML}{E69F00}
\definecolor{okabe3}{HTML}{56B4E9}
\definecolor{okabe4}{HTML}{009E73}
\definecolor{okabe5}{HTML}{F0E442}
\definecolor{okabe6}{HTML}{0072B2}
\definecolor{okabe7}{HTML}{D55E00}
\definecolor{okabe8}{HTML}{CC79A7}

\renewcommand{\emph}[1]{\textit{\textbf{#1}}}
\let\epsilon\varepsilon

\makeatletter
\newcommand{\drawCircleAlgorithm}[3][0.3]{%
  \@ifnextchar\bgroup{\@drawCircleAlgorithm{#1}{#2}{#3}}{\@drawCircleAlgorithm{#1}{#2}{#3}{}}%
}
\newcommand{\@drawCircleAlgorithm}[4]{%
  \def\circles{#2}
  \draw[line width=#1,dashed] (0,0) circle (1);
  \foreach \x/\y/\r [count=\i] in \circles {
    \pgfmathtruncatemacro\colorindex{mod(\i-1,7)+2}
    \draw[line width=#1,okabe\colorindex] (\x,\y) circle (\r);
    \pgfmathsetmacro\scalefactor{max(0.4, min(1.2, 2.5*\r))}
    \ifnum\i<#3
      \node[okabe\colorindex,scale=\scalefactor] at (\x,\y) {\i};
    \else\ifnum\i=#3
      \node[okabe\colorindex,scale=\scalefactor] at (\x,\y) {#4};
    \fi\fi
  }
}
\makeatother

\newcommand{\qdrresults}[3]{
  $\begin{array}{r@{\,}l}
      P(n) & \leq #1 \lceil \log{n} \rceil \\
      D(n) & \leq #2 n \\
      R_\text{max} & \leq #3 \lceil \log{n} \rceil
  \end{array}$%
}

\newcommand{\qdresults}[2]{
  $\begin{array}{r@{\,}l}
      P(n) & < #1 \lceil \log{n} \rceil \\
      D(n) & \leq #2 n
  \end{array}$%
}

\makeatletter
\DeclareRobustCommand{\qed}{%
  \ifmmode 
  \else \leavevmode\unskip\penalty9999 \hbox{}\nobreak\hfill
  \fi
  \quad\hbox{\qedsymbol}}
\newcommand{\openbox}{\leavevmode
  \hbox to.77778em{%
  \hfil\vrule
  \vbox to.675em{\hrule width.6em\vfil\hrule}%
  \vrule\hfil}}
\newcommand{\qedsymbol}{\openbox}
\newenvironment{proof}[1][\proofname]{\par
  \normalfont
  \topsep6\p@\@plus6\p@ \trivlist
  \item[\hskip\labelsep\itshape
    #1.]\ignorespaces
}{%
  \qed\endtrivlist
}
\newcommand{\proofname}{Proof}
\makeatother

\renewcommand{\emph}[1]{\textbf{\textit{#1}}}

\title{The Marco Polo Problem: \\
A Combinatorial Approach to Geometric Localization}
\author{
Ofek Gila\textsuperscript{1},
Michael T. Goodrich\textsuperscript{1},
Zahra Hadizadeh\textsuperscript{2},
Daniel S. Hirschberg\textsuperscript{1},
Shayan Taherijam\textsuperscript{1}\\
\textsuperscript{1}University of California, Irvine\\
\textsuperscript{2}University of Rochester\\
\{ogila,goodrich,dhirschb,staherij\}@uci.edu,
zhadizadeh99@gmail.com
}

\begin{document}

%
%
\pagestyle{plain}

\maketitle

\begin{abstract}
We introduce and study the \emph{Marco Polo problem},
which is a combinatorial approach to geometric localization.
In this problem, we are told there are
one or more points of interest (POIs)
within distance $n$ of the origin that we wish to localize.
Given a mobile search point, $\Delta$,
that is initially at the origin,
a localization algorithm is a strategy to move $\Delta$
to be within a distance of $1$ of a POI.
In the combinatorial localization problem we study, the only tool
we can use is reminiscent of the children's game,
``Marco Polo,'' in that $\Delta$ can issue a
\emph{probe} signal out a specified distance, $d$,
and the search algorithm
learns whether or not there is a POI within distance $d$ of $\Delta$.
For example, we could imagine that POIs are
one or more hikers lost in a forest and we need to design a
search-and-rescue (SAR) strategy to find
them using radio signal probes to a response device that hikers carry.
Unlike other known localization algorithms, probe responses do not inform
our search algorithm of the direction or distance to a POI.
The optimization problem
is to minimize the number of probes and/or POI responses, as well
as possibly minimizing the distance traveled by $\Delta$.
We describe a number of efficient combinatorial Marco Polo
localization strategies
and we analyze each one in terms of the size, $n$, of the search domain.
Moreover, we derive strong bounds for the constant factors for the search costs
for our algorithms, which in some cases involve computer-assisted proofs.
We also
show how to extend these strategies to find all POIs using a simple,
memoryless search algorithm, traveling a distance that is
$\mathcal{O}(\log{k})$-competitive with the optimal traveling salesperson (TSP)
tour for $k$ POIs.
\end{abstract}

\section{Introduction}
In the children's game, ``Marco Polo,'' 
a group of children are playing in a swimming pool.
One player is chosen as ``it,''
who closes their eyes and tries to find and tag one of the other players.
The ``it'' player periodically calls out ``Macro'' and the other players who
can hear this call must respond with ``Polo.''
The ``it'' player moves based on this ``Marco-Polo'' call-and-response
protocol until getting close enough to another player to tag them,
which ends this player's turn being ``it.''
See, e.g.,~\cite{wiki:marco}.

In this paper, we introduce and study the \emph{Marco Polo problem},
which is a combinatorial approach to geometric localization
motivated from the ``Marco Polo'' children's game.
In this problem, we start with a search point, $\Delta$, 
at the origin, with one or more points of interest (POIs) at a distance $n$
from the origin and our goal is to move $\Delta$ to be within distance
$1$ of a POI.\footnote{This formulation is made without loss of generality,
   as we could just as easily normalize the search
   problem so that there is a POI within distance $1$ of $\Delta$
   and we are interested in moving $\Delta$ to be within distance $\epsilon$
   of a POI, for a given $\epsilon>0$.}
We may periodically send probes out a specified distance, $d$, 
and we learn whether or not there is a POI within distance $d$ from $\Delta$.
Optimization goals include minimizing the number of probes and minimizing
the total distance traveled by $\Delta$.


We can motivate the Marco Polo problem, for example,
from a search-and-rescue (SAR) scenario.
Suppose a hiker, who we'll call ``Alice,'' 
is lost and stationary at a point of interest (POI)
in a large forest and we would like to find her
using an efficient SAR strategy.
Assume Alice has a wireless 
device, similar to an Apple AirTag, which
can respond to probes sent from a searcher, $\Delta$.
In particular, suppose $\Delta$ can send a 
probe at a given power level,
which sends an omni-directional signal out to a known radius (depending
on the power level), and if Alice
is present inside the circle determined by this probe, then $\Delta$
will receive a positive response.
Such probes use up power, however, both for $\Delta$ and for Alice's 
tracking device; hence, the goal is to devise 
a sequence of probes that minimizes the number of probes
needed to locate Alice to a specified accuracy.
(See Figure~\ref{fig:alice}.)

\begin{figure}[hbt!]
	\centering
	\resizebox*{0.6\linewidth}{!}{\begin{tikzpicture}
    \drawCircleAlgorithm[0.3]{
        {0.3240584078787189/0.468021961196084/0.8221566712745698},
        {-0.6536636388727417/-0.3403333825180739/0.675941592121281},
        {0.19819922134273138/-0.8073917009340685/0.5557298893544654},
        {-0.6588517518809158/0.597628201790139/0.4568970359594523},
        {0.8687374815033344/-0.32278331395764914/0.3756409461996407}%
    }{5}{\ding{55}}
\end{tikzpicture}}
	\vspace*{-\medskipamount}
	\caption{An example sequence of probes for an instance of 
                the Marco Polo problem. In this
		example, the first four probes are 
		negative and the POI (marked with an \ding{55}) is found 
                to within distance~$1$ on the fifth probe.
	}
	\label{fig:alice}
	\vspace*{-\bigskipamount}
\end{figure}
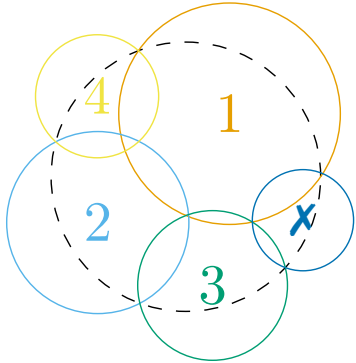

We also consider
generalizations of this problem, such as if there
were multiple POIs.
One can imagine other applications besides search and rescue
for the combinatorial searching problem,
including locating animals wearing tracking collars, 
finding radioactive sources, or identifying
anomalous readings in wireless sensor networks.

\subsection{Related Prior Work}
Although we are not familiar with any prior work
on the Marco Polo problem itself,
our work nevertheless falls into a rich area of study known
as \emph{localization algorithms}; see, e.g.,
the survey by Han, Xu, Duong, Jiang, and Hara~\cite{han2013localization}.
We discuss how our work compares to a wealth of existing prior work in an
appendix.

\subsection{Problem Definition}
In the \emph{Marco Polo problem},
there are $k\ge 1$ points of interest (POIs)
with unknown positions, with at least one that is within
a distance, $n$, of a point, $O$, called the \emph{origin},
which is the initial position of a mobile search point, $\Delta$. 
A \emph{probe}, $p(x,y,d)$, is a query that asks if there are any POIs
within distance $d$ of the current location, $(x,y)$, for $\Delta$.
The goal is to design a search strategy to 
move $\Delta$ to be within a distance of $1$ of a POI.
Given this setup, there are a number of possible constraints that 
define instances of the Marco Polo problem, including:
\begin{itemize}
\item
A search algorithm can be
\emph{incremental}, which finds all $k$ POIs one at a time,
or \emph{coordinated}, which finds multiple POIs in a coordinated
fashion.

\item
An algorithm is
\emph{memoryless} with respect to the current search
area if it%
s state is
restricted to the area determined by the previous successful probe.
\end{itemize}

There are a number of metrics we can 
use to measure the effectiveness of the strategy, including the following:
\begin{itemize}
\item
$P(n)$: the number of probes issued.
\item
$R_\text{max}$: the maximum number of times a POI responds to a probe.
\item
$D(n)$: the total distance traveled by $\Delta$.
\end{itemize}
These efficiency measures can conflict, of course, in that 
a strategy that minimizes, say, $R_\text{max}$, may have poor bounds for
$P(n)$ and/or $D(n)$.
Such a trade-off may nevertheless be worthwhile, however, such as 
in an SAR scenario where
the batteries are running out in a hiker's device and 
responses are costly.

\subsection{Our Results}
In this paper, we provide a number of efficient algorithms for solving 
instances of the Marco Polo problem.
We begin with a simple double binary-search warm-up algorithm that uses 
$2\lceil \log n\rceil + \mathcal{O}(1)$ probes,\footnote{All of the
    logarithms used in this paper are base 2.}
but which is based on unrealistic assumptions, as we discuss.
Instead, under more realistic assumptions about conditions
regarding the search space, we provide a sequence of algorithms, starting
with simple algorithms based on hexagon geometries, which we
call ``hexagonal algorithms,'' and progressing to more sophisticated
recursive strategies based on progressively shrinking probes at each
level of recursion.
This ultimately results in an algorithm that makes at most 
$3.34\lceil\log n\rceil$ probes using a monotonically spiraling search
strategy at each recursive level.
Using computer-assisted proof techniques, we then show that it is possible to 
find a POI using $2.53\lceil\log n\rceil$ probes. Although
this strategy
sacrifices the simplicity of a monotonically spiraling search,
it performs competitively with our
proven lower bound of $2.4 \lceil \log{n} \rceil$ probes for progressive
shrinking algorithms.
We also provide various algorithms that reach different trade-offs between the
number of probes made and the total distance traveled by $\Delta$, include one
algorithm that 
travels a total distance of at most $6.02 n$,
which is less than the circumference of the original search area.
We then provide a family of algorithms that are able to restrict the total
number of POI responses, $R_\text{max}$, to any desired value from 
$1$ to $\lceil \log{n} \rceil$ while minimizing the number of probes made.
Finally, we present a strategy to extend our incremental algorithms to find
all POIs while traveling a distance that is
$\mathcal{O}(\log{k})$-competitive with an optimal traveling salesperson (TSP)
tour.
We include experiments supporting all our results in an appendix.


\vspace{-\medskipamount}


\section{Finding One POI}

We describe a series of
progressively more efficient algorithms that minimize the number of probes
needed for finding a single POI.
We assume that there may be multiple POIs, either within the
search region of radius $n$ or outside of it, but we are initially
interested only in finding one of them.
Later we will discuss how to extend these algorithms to find all POIs.

\subsection{Hexagonal Algorithms}
Our first algorithms, which we call \emph{hexagonal algorithms},
are defined in terms of a tiling of
our search area with hexagons of radius $n / 2$.
There are seven such hexagons, which can each be 
probed with radius-$n / 2$
probes until a probe succeeds, which then allows us
to make a recursive call to an $n/2$-sized subproblem.
Such an algorithm will
take $7 \lceil \log{n} \rceil$ probes in the
worst case where POIs are always in the last hexagon probed
in each recursive level.
We can improve this
to $6 \lceil \log{n} \rceil$ probes by not probing the last hexagon,
since a POI \emph{must} be there if the other six probes fail.
We refer to this hexagonal algorithm as Algorithm 1.

There is a better hexagonal algorithm, however, which involves
first probing the upper two quadrants, which can be done with radius $n /
\sqrt{2}$ probes, eliminating 3 hexagons, and then probing 3 of the 4 remaining
hexagons as before.
We refer to this modified hexagonal algorithm as
Algorithm 2. See \Cref{fig:alg-1-2}.
If POIs are in one of the two larger probes, we reduce the problem less than
before, only by a factor of $\sqrt{2}$, but with fewer probes (at most 2).
This turns out to be a better tradeoff, so that in the worst case where POIs
are all in the last hexagon probed at each level
and
our algorithm makes at most $5 \lceil \log{n} \rceil$ probes.

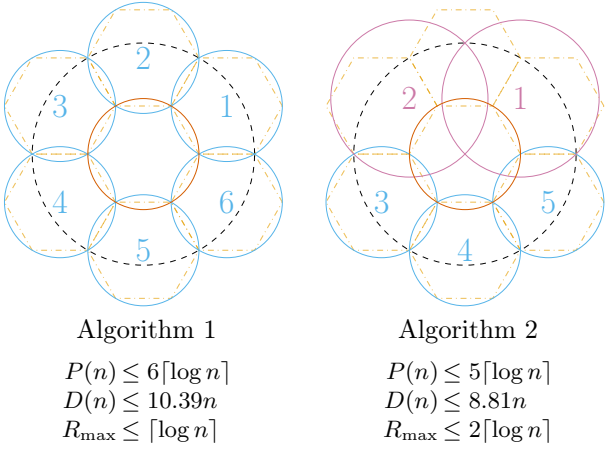
\begin{figure}[hbt!]
	\vspace*{-\smallskipamount}
	\captionsetup[subfigure]{labelformat=empty}
	\centering
	\begin{subfigure}{0.47\linewidth}
		\centering
		\resizebox*{\linewidth}{!}{\begin{tikzpicture}
	\newcommand{\drawhexagon}[3]{
		\draw[shift={(#1,#2)},okabe2,dashdotted,opacity=0.4] (0:#3) \foreach \x in {60,120,...,360} {  -- (\x:#3) };
	}

	\draw[dashed] (0,0) circle (2);

	\def\hexradius{1}
	\def\gridradius{1}  

	\foreach \ring in {1,...,\gridradius} {
		\foreach \side in {0,...,5} {
			\foreach \pos in {0,...,\ring} {
				\pgfmathsetmacro\angle{60*\side + 30}
				\pgfmathsetmacro\x{(\ring-\pos)*\hexradius*sqrt(3)*cos(\angle) + \pos*\hexradius*sqrt(3)*cos(\angle+60)}
				\pgfmathsetmacro\y{(\ring-\pos)*\hexradius*sqrt(3)*sin(\angle) + \pos*\hexradius*sqrt(3)*sin(\angle+60)}

				\drawhexagon{\x}{\y}{\hexradius}
			}
		}
	}

	\foreach \side in {0,...,5} {
		\pgfmathsetmacro\angle{60*\side + 30}
		\pgfmathsetmacro\x{\hexradius*sqrt(3)*cos(\angle)}
		\pgfmathsetmacro\y{\hexradius*sqrt(3)*sin(\angle)}

		\draw[okabe3] (\x,\y) circle (1);
		\node[okabe3] at (\x,\y) {\LARGE \the\numexpr\side+1};
	}

	\draw[okabe7] (0,0) circle (1);
\end{tikzpicture}}
		Algorithm 1
		\caption{\qdrresults{6}{10.39}{}}
	\end{subfigure}
	\hfill
	\begin{subfigure}{0.47\linewidth}
		\centering
		\resizebox*{\linewidth}{!}{\begin{tikzpicture}
	\newcommand{\drawhexagon}[3]{
		\draw[shift={(#1,#2)},okabe2,dashdotted,opacity=0.4] (0:#3) \foreach \x in {60,120,...,360} {  -- (\x:#3) };
	}

	\draw[dashed] (0,0) circle (2);

	\def\hexradius{1}
	\def\gridradius{1}  

	\foreach \ring in {1,...,\gridradius} {
		\foreach \side in {0,...,5} {
			\foreach \pos in {0,...,\ring} {
				\pgfmathsetmacro\angle{60*\side + 30}
				\pgfmathsetmacro\x{(\ring-\pos)*\hexradius*sqrt(3)*cos(\angle) + \pos*\hexradius*sqrt(3)*cos(\angle+60)}
				\pgfmathsetmacro\y{(\ring-\pos)*\hexradius*sqrt(3)*sin(\angle) + \pos*\hexradius*sqrt(3)*sin(\angle+60)}

				\drawhexagon{\x}{\y}{\hexradius}
			}
		}
	}

	\foreach \side in {3,...,5} {
		\pgfmathsetmacro\angle{60*\side + 30}
		\pgfmathsetmacro\x{\hexradius*sqrt(3)*cos(\angle)}
		\pgfmathsetmacro\y{\hexradius*sqrt(3)*sin(\angle)}

		\draw[okabe3] (\x,\y) circle (1);
		\node[okabe3] at (\x,\y) {\LARGE \side};
	}

	\pgfmathsetmacro\st{sqrt(2)}
	\draw[okabe8] (1,1) circle (\st);
	\node[okabe8] at (1,1) {\LARGE 1};
	\draw[okabe8] (-1,1) circle (\st);
	\node[okabe8] at (-1,1) {\LARGE 2};

	\draw[okabe7] (0,0) circle (1);
\end{tikzpicture}}
		Algorithm 2
		\caption{\qdrresults{5}{8.81}{2}}
	\end{subfigure}
	\caption{\label{fig:alg-1-2}
		Two simple hexagonal algorithms.
Algorithm 1 performs probes of
		radius $n / 2$ along the center of each of the 6 outer hexagons, while
		Algorithm 2 first performs two radius $n / \sqrt{2}$ probes in
		the upper two quadrants before performing the remaining probes of Algorithm
		1.
}
\vspace*{-\medskipamount}
\end{figure}

\subsubsection{Distance Traveled}
To determine the total distance traveled by $\Delta$, $D(n)$,
in any probe-sequence algorithms,
including our hexagonal algorithms,
let us consider the
first layer of probes, i.e., those performed when the search radius is~$n$.
In particular,
consider the $k$-th probe (starting from 1)
as having radius $r_k = \rho_k n$,
where $\rho_k$ represents the proportionality factor for the size of each probe
relative to the (current) search radius.
Let $d_k$ denote the distance $\Delta$ travels to perform the $k$-th probe.
Assuming we always first succeed on our $k$-th probe, the total distance
traveled is determined by
$D(n) = d_k + D(\rho_k n)$,
where we assume that $D(n) = b_k n$, for some constant $b_k$.
Solving this, we find that if indeed the first probe that
succeeds is always the $k$-th probe, then $b_k = d_k / (1 - \rho_k)$.
Using this, we determine that:
\begin{equation}\label{eq:distance}
  D(n) \leq \max_k{\frac{d_k}{1 - \rho_k}} n.
\end{equation}
To determine a good bound for this maximum, we use a computer-assisted
proof technique to compute this value
for each algorithm.\footnote{We provide pseudocode for this algorithm in an appendix.}
%
The maximum number of responses is achieved when the POI is always inside the
largest probe, $\rho_\text{max}$, resulting in:
\begin{equation}\label{eq:max_responses}
	R_\text{max} \leq \lceil \log_{1/\rho_\text{max}}{n} \rceil = -\frac{1}{\log{\rho_\text{max}}} \lceil \log{n} \rceil.
\end{equation}

\subsection{Progressively Shrinking Probes}
As shown in Algorithm 2, we may be able to achieve better results by probing
from differently sized circles before ever receiving a positive response.
Indeed, we can do better by using
probes that get progressively smaller, so that if we spend many
probes to be able to recurse to a smaller search area,
we should at least reduce the remaining area by a larger factor.
The question remains how to choose the sizes of the probes.

Let the total number of probes required to find a POI starting with a search
area of radius $n$ be $P(n)$, and let the $k$-th probe (starting from 1) have
radius $r_k = \rho_k n$, where $\rho_k$ represents the proportionality factor for the
size of each probe relative to current search radius.
Assuming we always first succeed on our $k$-th probe, 
$P(n)$ is
determined by the recurrence relation,
$P(n) = k + P(\rho_k n)$,
where
$P(n) = c_k \lceil \log{n} \rceil$, for some constant, $c_k>0$
, and
our total number of probes will be
$k + P(\rho_k n)$, 
resulting in
\begin{align*}
	P(n) &= k + P(\rho_k n) = k + c_k \lceil \log(\rho_k n) \rceil \\
	&= k + c_k \log{\rho_k} + c_k \lceil \log{n} \rceil \\
	&= k + c_k \log{\rho_k} + P(n).
\end{align*}
Hence, $0 = k + c_k \log{\rho_k} \implies c_k = -\frac{k}{\log{\rho_k}}.$

\medskip
In the worst case, all POIs will be in the
$k$-th probe such
that the total number of probes is maximized, i.e., such that $c_k$ is
maximized, so we pick $c_k$ such that it is the same for all $k$.
This is done by setting $\rho_k = \rho_1^k$, where $\rho_1$ is the proportionality
constant of the first probe, resulting in the overall number of probes
being\footnote{The inequality arises from the omission of the last probe.}
\begin{equation}\label{eq:prog-queries}
	P(n) \leq -\frac{1}{\log{\rho_1}} \lceil \log{n} \rceil.
\end{equation}

Note that our equations for the maximum number of probes and maximum number of
responses, \Cref{eq:prog-queries,eq:max_responses}, respectively, are the
same---the latter occurring when the POI is always in the first probe.


\subsubsection{A Lower Bound for Progressive Shrinking}
With this result in mind, we can determine a lower bound on the number of probes
required to find a POI.
Because $\Delta$'s algorithm is memoryless with regards to its current search
area, it is always possible for POIs to be along the area's perimeter, so
any algorithm must at least probe the perimeter of the search area.
To maximize the perimeter coverage of each probe, we place it such that its
diameter is a chord of the circle, as shown in \Cref{fig:perimeter}, and
determine the minimum value of $\rho_1$ required to probe the entire perimeter of
the circle to be approximately 0.74915.
This results in a lower bound of $P(n) > 2.40001 \lceil \log{n} \rceil$
probes.\footnote{The exact coefficient $c$ can be determined numerically by
solving the following equation:
$\sum_{k=1}^\infty{\sin^{-1}{2^{-\frac{k}{c}}}} = \pi$, which we approximated
using Wolfram Mathematica.}

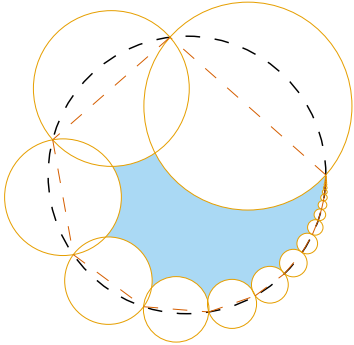
\begin{figure}[hbt!]
	\centering
	\resizebox*{0.6\linewidth}{!}{\begin{tikzpicture}
	\tikzset{
		clip even odd rule/.code={\pgfseteorule}, 
		invclip/.style={
			clip,insert path=
				[clip even odd rule]{
					[reset cm](-\maxdimen,-\maxdimen)rectangle(\maxdimen,\maxdimen)
				}
		}
	}

	\draw[line width=0.3,dashed] (0,0) circle (1);

	\def\po{0.749154715}
	
	\def\radius{\po}
	\def\figuretheta{0}

	\def\numcircles{25}

	\begin{scope}
		\clip (0,0) circle (1);

		\foreach \i in {1,...,\numcircles} {
			\pgfmathsetmacro\xo{cos(\figuretheta)}
			\pgfmathsetmacro\yo{sin(\figuretheta)}

			\pgfmathparse{\figuretheta+2*asin(\radius)}
			\global\let\figuretheta\pgfmathresult

			\pgfmathsetmacro\xt{cos(\figuretheta)}
			\pgfmathsetmacro\yt{sin(\figuretheta)}

			\pgfmathsetmacro\x{(\xo + \xt)/2}
			\pgfmathsetmacro\y{(\yo + \yt)/2}

			\clip[invclip] (\x,\y) circle (\radius);
			
			\pgfmathparse{\radius*\po}
			\global\let\radius\pgfmathresult
		}

		\fill[okabe3!50] (0,0) circle (1);
	\end{scope}

	\def\radius{\po}
	\def\figuretheta{0}

	\foreach \i in {1,...,\numcircles} {
		\pgfmathsetmacro\xo{cos(\figuretheta)}
		\pgfmathsetmacro\yo{sin(\figuretheta)}

		\pgfmathparse{\figuretheta+2*asin(\radius)}
		\global\let\figuretheta\pgfmathresult

		\pgfmathsetmacro\xt{cos(\figuretheta)}
		\pgfmathsetmacro\yt{sin(\figuretheta)}

		\pgfmathsetmacro\x{(\xo + \xt)/2}
		\pgfmathsetmacro\y{(\yo + \yt)/2}

		\draw[okabe2,line width=0.1] (\x,\y) circle (\radius);
		\draw[okabe7,line width=0.2,dashed] (\xo,\yo) to (\xt,\yt);
		
		\pgfmathparse{\radius*\po}
		\global\let\radius\pgfmathresult
	}
\end{tikzpicture}}
	\caption{\label{fig:perimeter}
		The optimal placement of progressively decreasing probes of radius
		proportional to $\rho_1^k$ in order to cover the circumference of the
		search area.
		$\rho_1$ must be at least 0.74915 to fully cover the perimeter, as shown.
		There exists uncovered search area, depicted in blue.
	}
	\vspace*{-\medskipamount}
\end{figure}

\subsubsection{Chord-Based Shrinking Algorithms}
The placement strategy for $\Delta$ used to show the lower bound of
roughly $2.4 \lceil \log{n} \rceil$ probes for progressive shrinking algorithms
is not a valid approach for an upper-bound for an algorithm, due to all
the internal uncovered area, see \Cref{fig:perimeter}.
Nevertheless, the approach of
placing diameters of each probe as chords of a
circle can work with large enough probes,
which we can determine by tuning $\rho_1$.
We refer to our next algorithm, which is a progressive shrinking algorithm,
as Algorithm 3.
In this algorithm, we numerically determine, using a computer-assisted proof,
that the minimum value of $\rho_1$ that leaves
no uncovered area is approximately 0.844,
reducing the number of probes to $P(n)  < 4.08 \lceil \log{n} \rceil$.
See \Cref{fig:alg-3-4} (left).
$\Delta$'s flight path maintains the property that it is monotonic in a counterclockwise orientation.

However, if we allow for nonmonotonic flight plans,
we can place the two largest probes side by side and alternate the next
two probes on either side, which leads to significant
improvements,
which we refer to as Algorithm 4.
This algorithm is able to further reduce the number of probes to $P(n)  < 3.54
\lceil \log{n} \rceil$, and is depicted in \Cref{fig:alg-3-4} (right).

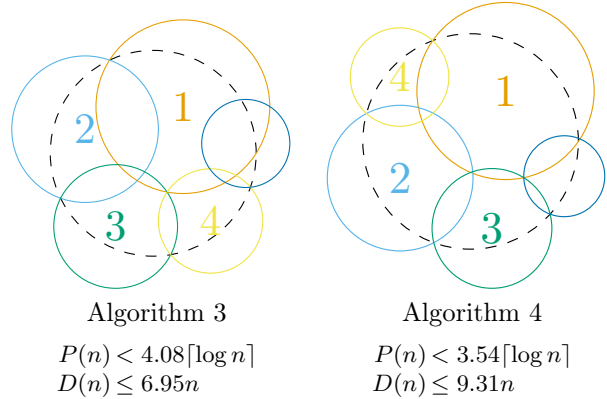
\begin{figure}[hbt!]
	\captionsetup[subfigure]{labelformat=empty}
	\centering
	\begin{subfigure}{0.48\linewidth}
		\centering
		\resizebox*{\linewidth}{!}{\begin{tikzpicture}
    \drawCircleAlgorithm[0.3]{
        {0.28789865898868494/0.4527836361234821/0.843860972560833},
        {-0.6618325909742372/0.23428465945130714/0.712101341011315},
        {-0.36476973938250307/-0.7112276461444245/0.6009145301876817},
        {0.5569297259284322/-0.6577923047806103/0.5070883198701132},
        {0.8987720218026585/0.09539463542963203/0.42791204277983247}%
    }{5}
\end{tikzpicture}}
		Algorithm 3
		\caption{\qdresults{4.08}{6.95}}
	\end{subfigure}
	\hfill
	\begin{subfigure}{0.48\linewidth}
		\centering
		\resizebox*{\linewidth}{!}{\begin{tikzpicture}
    \drawCircleAlgorithm[0.3]{
        {0.3240584078787189/0.468021961196084/0.8221566712745698},
        {-0.6536636388727417/-0.3403333825180739/0.675941592121281},
        {0.19819922134273138/-0.8073917009340685/0.5557298893544654},
        {-0.6588517518809158/0.597628201790139/0.4568970359594523},
        {0.8687374815033344/-0.32278331395764914/0.3756409461996407}%
    }{5}
\end{tikzpicture}}
		Algorithm 4
		\caption{\qdresults{3.54}{9.31}}
	\end{subfigure}
	\caption{\label{fig:alg-3-4}
		Algorithms 3 and 4 both place probes such that their diameters are
		chords of the search area circle.
		Algorithm 3 places the probes in order of decreasing size going
		counter-clockwise, while Algorithm 4 places the probes such that
		they overlap as little as possible.
	}
	\vspace*{-\bigskipamount}
\end{figure}

An optimization can be made to the total distance traveled by $\Delta$
for Algorithms~3 and~4.
If a POI is determined to be in the last area, because previous probes
in a recursive level are all negative, $\Delta$ 
can travel directly to the center of the first probe within the next recursive
layer instead of the center of the last probe, since the last probe
does not need to be performed.
Accordingly,
Algorithms~3 and~4 tradeoff the number of probes
for flight distance.
See \Cref{fig:alg-3-4}.

\subsubsection{Higher-count Monotonic-path Algorithms}
Both Algorithms 3 and 4 use very few probes (i.e., 5) at each
recursive layer,
while we know from \Cref{eq:prog-queries} that we can introduce more probes with
geometrically decreasing radii without increasing $P(n)$.
However, decreasing the value of $\rho_1$ in Algorithm 4 will not only introduce a
gap in the perimeter
but would also introduce an internal gap, which would require
another placement scheme entirely to fill.
Our next idea is to begin with one large central probe before placing the
remaining probes along the perimeter monotonically as in Algorithm~3, which
should also have good flight-path performance for $\Delta$.
Intuitively,
the large central probe greatly reduces the probe radius required to avoid
internal gaps, allowing for more probes to be placed.
Indeed, by placing the remaining probes such that their diameters are
chords of the search area, as was done in Algorithm 3, leads to Algorithm~5,
which uses up to 8 probes at each recursive level.
While this algorithm improves upon Algorithm 3, requiring only $P(n)  < 3.83
\lceil \log{n} \rceil$ probes, and improved flight distance,
it performs worse than Algorithm 4 in its total number of probes.
%
See Figure~\ref{fig:alg-5-6}.
In order to take advantage of even more probes, we observe that we cover the
circumference of the search area at a much faster rate than the circumference of
the central probe.
In other words, if we reduce the probe radii, we would still be able to cover
the search radius circumference, but we would introduce internal gaps between
the outer probes and the central probe.
Ideally, we would like the rate at which they cover the inner and the outer
circumferences to be the same, such that the
chords made with the outer and inner circles cover the same angle.
%
The geometric reasoning
is shown in
\Cref{fig:angle}.

\begin{figure}[hbt!]
	\centering
	\resizebox*{0.5\linewidth}{!}{\begin{tikzpicture}[
	every node/.style={scale=0.3},
	line width=0.1,
	label distance=-3pt,
	bl/.style={label={below left:#1}},
	al/.style={label={[label distance=-7pt]above left:#1}},
	l/.style={label={left:#1}},
	br/.style={label={below right:#1}},
	b/.style={label={below:#1}},
	sbl/.style={label={[xshift=5pt]below left:#1}},
	ar/.style={label={above right:#1}},
]

	\begin{scope}
		\clip (-0.15, -0.15) rectangle (1.1, 0.705);

		\pgfmathsetmacro\rz{1}

		\pgfmathsetmacro\ro{0.7*\rz}

		\node[bl=$O$] (O) at (0, 0) {};

		\draw[dashed] (O) -- node[midway, left] {$r_1$} (0, \ro);

		\def\starttheta{5}

		\pgfmathsetmacro\xA{\ro*cos(\starttheta)}
		\pgfmathsetmacro\yA{\ro*sin(\starttheta)}

		\node[bl=$A$] (A) at (\xA, \yA) {};

		\pgfmathsetmacro\xAp{\rz*cos(\starttheta)}
		\pgfmathsetmacro\yAp{\rz*sin(\starttheta)}

		\node[sbl=$A'$] (Ap) at (\xAp, \yAp) {};

		\draw[okabe2] (O.center) -- (Ap.center);

		\def\figtheta{18}

		\pgfmathsetmacro\xB{\ro*cos(\starttheta+2*\figtheta)}
		\pgfmathsetmacro\yB{\ro*sin(\starttheta+2*\figtheta)}

		\node[l=$B$] (B) at (\xB, \yB) {};

		\pgfmathsetmacro\xBp{\rz*cos(\starttheta+2*\figtheta)}
		\pgfmathsetmacro\yBp{\rz*sin(\starttheta+2*\figtheta)}

		\node[l=$B'$] (Bp) at (\xBp, \yBp) {};

		\draw[okabe2] (O.center) -- (Bp.center);

		\pgfmathsetmacro\xM{(\xA + \xAp)/2}
		\pgfmathsetmacro\yM{(\yA + \yAp)/2}

		\node[sbl=$M$] (M) at (\xM, \yM) {};

		\pgfmathsetmacro\dM{sqrt(\xM*\xM + \yM*\yM)}
		\pgfmathsetmacro\dOp{\dM/cos(\figtheta)}

		\pgfmathsetmacro\xOp{\dOp*cos(\starttheta+\figtheta)}
		\pgfmathsetmacro\yOp{\dOp*sin(\starttheta+\figtheta)}

		\node[al=$P_k$] (Op) at (\xOp, \yOp) {};

		\pgfmathsetmacro\rk{sqrt((\rz-\ro)/2*(\rz-\ro)/2 + \dM*\dM*tan(\figtheta)*tan(\figtheta))}

		\draw[okabe3] (Op) circle (\rk);

		\draw[okabe2] (O.center) -- (Op.center);

		\draw[okabe3,dashed] (A.center) -- node[midway, left, xshift=5pt] {$r_k$} (Op.center);

		\draw[okabe2] (M.center) -- (Op.center);

		\draw pic[draw, angle radius=0.3cm, angle eccentricity=1.5, "$\theta$", okabe1] {angle = A--O--Op};

		\draw pic[draw, angle radius=0.2cm, angle eccentricity=1.5, okabe1] {right angle = A--M--Op};


		\draw (O) circle (\rz);

		\draw (O) circle (\ro);

		\fill (O) circle (0.02);
		\fill (A) circle (0.02);
		\fill (Ap) circle (0.02);
		\fill (B) circle (0.02);
		\fill (Bp) circle (0.02);
		\fill (M) circle (0.02);
		\fill (Op) circle (0.02);
	\end{scope}
\end{tikzpicture}}
	\caption{\label{fig:angle}
		A diagram showing the relationship between the $k$-th circle,
		centered at $P_k$, the first circle, centered at $O$, and the search
		radius, also centered at $O$.
	}
	\vspace*{-\medskipamount}
\end{figure}
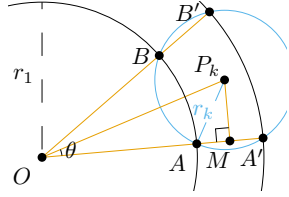

In particular,
we must determine at what position, $P_k$, to place the center $C_k$ of the
$k$-th probe with radius $r_k$.
For simplicity, we assume that the search area is a unit circle centered at the
origin, $O$, and that the first probe, $C_1$, has radius $r_1 < 1$.
Let $A$, $B$, $A'$, and $B'$ refer to the points where $C_k$ intersects $C_1$
and the search area, respectively, and let $M$ be the midpoint of
$\overline{AA'}$.
Covering both the inner and outer circumferences at the same rate implies that
$\angle A O P_k = \angle A' O P_k$, i.e., that $\overline{OA}$ and
$\overline{OA'}$ are colinear.
Since $\triangle A A' P_k$ is isosceles, we know that
%
%
%
%
$\overline{P_k M}$ is perpendicular to $\overline{AA'}$ so $|\overline{P_kM}| = \sqrt{r_k^2 -
(\frac{1-r_1}{2})^2}$ and consequently $\theta = \arctan(\frac{\sqrt{r_k^2 -
(\frac{1-r_1}{2})^2}}{1-\frac{1-r_1}{2}})
$.
Thus, by moving the outer probes inward, towards the origin, to their new centers
as described above, we are
able to ``turn it up to 11'' and achieve Algorithm 6.
See Figure~\ref{fig:alg-5-6}.

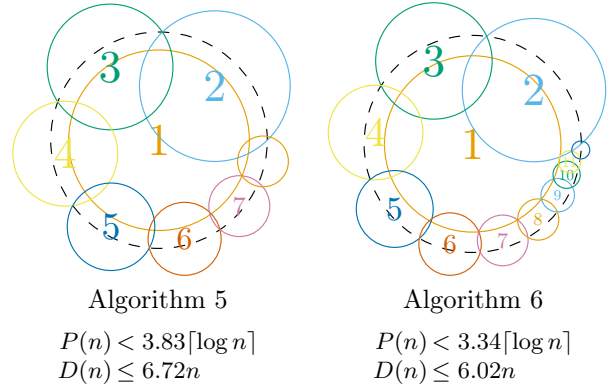
\begin{figure}[htb!]
	\captionsetup[subfigure]{labelformat=empty}
	\centering
	\begin{subfigure}{0.48\linewidth}
		\centering
		\resizebox*{\linewidth}{!}{\begin{tikzpicture}
    \drawCircleAlgorithm[0.3]{
        {0/0/0.8342831628142449},
        {0.5155444723022481/0.4997583109672561/0.6960283957553398},
        {-0.45191758342842164/0.6771821719669994/0.5806847714192898},
        {-0.8658878771785593/-0.12466365083014752/0.4844555276977519},
        {-0.44332373517785195/-0.8000676091658153/0.4041730898905245},
        {0.23657851953538434/-0.9112245982264008/0.3371948037582729},
        {0.7415164177729718/-0.609109793002744/0.2813159473639805},
        {0.9600399845076252/-0.19586529211226786/0.23469715831690727}%
    }{8}
\end{tikzpicture}}
		Algorithm 5
		\caption{\qdresults{3.83}{6.72}}
	\end{subfigure}
	\hfill
	\begin{subfigure}{0.48\linewidth}
		\centering
		\resizebox*{\linewidth}{!}{\begin{tikzpicture}
    \drawCircleAlgorithm[0.3]{
        {0/0/0.8124407481991511},
        {0.5643208368822755/0.4958455706596195/0.6600599693343965},
        {-0.3572231106256195/0.7647334661094909/0.5362596153423458},
        {-0.8939478838207515/0.10507543879251385/0.4356791631177244},
        {-0.7179732867681129/-0.599353148346416/0.3539635052581441},
        {-0.2082732502738843/-0.9229095419396449/0.28757437504712074},
        {0.27639770841541267/-0.8891745432525523/0.23363714042618605},
        {0.6027932264816553/-0.6964993285330394/0.18981633317496074},
        {0.782929757696528/-0.4724868908724668/0.15421452374508446},
        {0.8646018242986668/-0.2839050168222197/0.12529016305463217},
        {0.8920706471869477/-0.16435352894045435/0.10179083381409901},
        {0.9950262135480192/-0.05553112863558944/0.08269902118374205}%
    }{12}
\end{tikzpicture}}
		Algorithm 6
		\caption{\qdresults{3.34}{6.02}}
	\end{subfigure}
	\caption{\label{fig:alg-5-6}
		Algorithms 5 and 6 both perform
		probes counter-clockwise along the
		circumference of a central probe.
		Algorithm 5 places the probes so that diameters are
		chords of the search area, while
		Algorithm 6 
		balances the coverage rate of the inner and outer circumferences.
		Out of all our algorithms, $\Delta$ travels the least
		distance (in the worst case) using Algorithm 6.
	}
	\vspace*{-\medskipamount}
\end{figure}

\subsubsection{Darting Non-Monotonic Algorithms}
Up to this point, our best two algorithms for minimizing
$P(n)$ are Algorithm~4, which is non-monotonic and uses its
few probes very efficiently, and Algorithm~6,
which is monotonic with a counterclockwise spiral of probes such that,
despite having significant overlap, it is able to squeeze in nearly three
times as many probes and achieve better performance.
The question remains whether it is possible to achieve even better query
performance at the expense of monotonicity, giving up on optimizing
$\Delta$'s flight distance so as to achieve even better probe complexity.
We refer to such algorithms as being \emph{darting algorithms}
and discuss them in an appendix.
The best two darting algorithms are highlighted in
\Cref{fig:alg-5-6}.

\begin{figure}[hbt!]
	\captionsetup[subfigure]{labelformat=empty}
	\centering
	\begin{subfigure}{0.48\linewidth}
		\centering
		\resizebox*{\linewidth}{!}{\begin{tikzpicture}
    \drawCircleAlgorithm[0.1]{
        {0.3766095638275147/0.4845356542723805/0.78955078125},
        {-0.6275047422486887/-0.46649990626374743/0.6233904361724854},
        {0.22055685720357815/-0.8420780260559434/0.49219840590376407},
        {-0.8363946658436578/0.38655122619585/0.38861563591132153},
        {0.8405974221134691/-0.4463745438181471/0.30683177893974944},
        {-0.4600705903487553/0.8541928925132176/0.2422592707742065},
        {0.08366242225629675/-0.23730948262677515/0.19127599650483001},
        {-0.3797354768401737/0.1731748713449593/0.15102211247476083},
        {0.4914502254478055/-0.35144440898159507/0.11923962689047278},
        {-0.4652119846895424/0.5499473330977886/0.09414574056733128},
        {0.9318232713465783/-0.12293596444208141/0.07433284301629624},
        {0.3311617259784244/-0.3349595104849733/0.0586895542760503},
        {-0.43723848007056154/0.34499464742122243/0.04633838342986979},
        {-0.4359177913258304/0.4247192866145801/0.03658650683891575},
        {0.983269323540918/-0.06016623510468254/0.028886905057874397},
        {0.2576485140468424/-0.3384791427308925/0.022807678456339308},
        {0.9872487090293736/-0.022826769121943327/0.018007820343701495},
        {0.27047360937857473/-0.30859021040107315/0.014218088620979157},
        {-0.4184952978014185/0.46201937613118615/0.01122590297857583},
        {-0.41254822050693163/0.38983550014736806/0.008863420466971248},
        {1.0018747269313437/-0.00674233527601123/0.006998120554244389},
        {-0.40476827646332914/0.378147326534274/0.00552537155088534},
        {-0.4220866421019454/0.3928031992667419/0.004362561424698044},
        {-0.43134390662214/0.46376467139943706/0.003444463781121454},
        {0.2816293232336179/-0.2991930013893405/0.002719579069371773},
        {-0.4086454225500263/0.3996048490523598/0.0021472457788936313}%
    }{26}
\end{tikzpicture}}
		Algorithm 7
		\caption{\qdresults{2.93}{25.8}}
	\end{subfigure}
	\hfill
	\begin{subfigure}{0.48\linewidth}
		\centering
		\resizebox*{\linewidth}{!}{\begin{tikzpicture}
    \drawCircleAlgorithm[0.1]{
        {0.31890822598683777/0.338730967237933/0.7606738281250001},
        {-0.5717205105126045/-0.3295360766189852/0.5786246727943422},
        {0.2358896072860885/-0.7299812819467005/0.4401446449020478},
        {-0.6641982506120474/0.4760614019924518/0.3348065119663595},
        {0.7748605073639385/-0.4180019802783681/0.2546785511386293},
        {-0.38010699353173694/0.8400219715055867/0.19372730843594976},
        {-0.24264202950615982/-0.9162668190962578/0.14736329332032652},
        {-0.9447256060355724/0.19353492303514935/0.11209540045508003},
        {0.9564496604536843/-0.16724250217220732/0.0852680373793706},
        {-0.17379964403490622/0.9601146741154272/0.06486116441007143},
        {0.6965819800911245/-0.6900956647420202/0.049338190228454044},
        {0.9807288035330222/-0.057222601157597974/0.03753027003383761},
        {-0.4053562184114525/-0.8968183725429388/0.028548294177204232},
        {-0.09750163160280584/0.9855966208641711/0.02171594021821259},
        {-0.9896472520894173/0.07824473838396612/0.01651874737712142},
        {-0.5759093376545397/0.8075517569129109/0.012565378803184755},
        {0.9973804732648248/-0.015560814119356059/0.00955815479605928},
        {-0.07195916029317578/0.9948692359094679/0.0072706381985297415},
        {-0.4361037554320548/-0.896194640256324/0.0055305841913874726},
        {-0.5902089700074181/0.8047294030747743/0.004206970648630317},
        {-0.44349277695179995/-0.8947145223662323/0.0032001324681031375},
        {0.9987003913562287/-0.004291986905113787/0.0024342570150191183},
        {-0.06365138400382925/0.9970901851237663/0.0018516756022547286},
        {-0.9972462458793971/0.06274660402702736/0.0014085211688127694},
        {0.6764337025042138/-0.735921194086661/0.0010714251894759087},
        {0.9997523067012348/-0.0014552409274862232/0.000815005100428193},
        {-0.06151462292102633/0.9978192624424138/0.0006199530496841137},
        {-0.5945955566037345/0.8037645786621477/0.0004715820595609831},
        {-0.060637739445547234/0.9980488218135815/0.00035872013052132485},
        {0.9998646003020775/-0.0004430020966700865/0.00027286901490915585},
        {0.5211519793965814/-0.39467171021677755/0.00020756431814764528},
        {-0.4468097834519504/-0.8945818241012953/0.00015788874446752475},
        {0.9999950750352506/-8.911522609467148e-05/0.00012010183567196199}%
    }{33}
\end{tikzpicture}


}
		Algorithm 8
		\caption{\qdresults{2.53}{45.4}}
	\end{subfigure}
	\caption{\label{fig:alg-7-8}
		Algorithms~7 and~8 both use computer-assisted probe placement to
		efficiently cover the search area. Algorithm~7 begins with
		Algorithm~4,
		removing the final probe, while Algorithm~8 uses a differential
		evolution algorithm to
		place the initial six probes.
		The remaining probes are placed greedily as discussed in an appendix.
		Algorithm 8 achieves our best probe results.
	}
	\vspace*{-\medskipamount}
\end{figure}

\begin{theorem}
Progressive shrinking algorithms for $\Delta$ searching a circular region
of radius $n$
have a lower bound of $2.4\lceil\log n\rceil$ for $P(n)$, and
we can achieve upper bounds as shown in Figures~\ref{fig:alg-1-2},
\ref{fig:alg-3-4},
\ref{fig:alg-5-6},
and~\ref{fig:alg-7-8}.
\end{theorem}

\subsubsection{Reducing The POIs' Responses}
Our exploration has so far focused on reducing the total number of
probes, $P(n)$, required to find a POI.
In this section, we consider the case where POIs
are limited in their number of
responses $R_\text{max}$ to probes, e.g., due to battery constraints.
Namely, consider a scenario where $R_\text{max}$ is a fixed value where
$R_\text{max} \ge 1$.
Our goal is to design a search strategy that minimizes the number of
probes, $P(n)$, while ensuring that the number of responses from the POIs
does not exceed $R_\text{max}$.
Recalling from \Cref{eq:max_responses}, the amount of responses is determined by
the size of the largest probe at each recursive layer, so there is no benefit
to using
differently sized probes.
We present a family of hexagonal algorithms which probe the search space using
hexagonal lattices with $L$ layers of hexagonal rings.
See \Cref{fig:hex-lattice}.

\begin{figure}[hbt!]
    \centering
    \resizebox*{0.5\linewidth}{!}{\begin{tikzpicture}
    \newcommand{\drawtophexagon}[3]{
        \draw[shift={(#1,#2)},okabe2,dashdotted] (0:#3) \foreach \x in {60,120,180} { -- (\x:#3) };
    }

    \newcommand{\drawbottomleftside}[3]{
        \draw[shift={(#1,#2)},okabe2,dashdotted] (180:#3) -- (240:#3);
    }

    \newcommand{\drawbottomhexagon}[3]{
        \draw[shift={(#1,#2)},okabe2,dashdotted] (240:#3) -- (300:#3);
    }

    \newcommand{\drawbottomrightside}[3]{
        \draw[shift={(#1,#2)},okabe2,dashdotted] (300:#3) -- (360:#3);
    }

    \def\hexradius{1}
    \def\gridradius{3}  
    \def\centerX{0}
    \def\centerY{0}



    \drawtophexagon{\centerX}{\centerY}{\hexradius}

    \foreach \ring in {1,...,\gridradius} {
        \foreach \side in {0,...,5} {
            \foreach \pos in {0,...,\ring} {
                \pgfmathsetmacro\angle{60*\side + 30}
                \pgfmathsetmacro\x{(\ring-\pos)*\hexradius*sqrt(3)*cos(\angle) + \pos*\hexradius*sqrt(3)*cos(\angle+60)}
                \pgfmathsetmacro\y{(\ring-\pos)*\hexradius*sqrt(3)*sin(\angle) + \pos*\hexradius*sqrt(3)*sin(\angle+60)}


                \drawtophexagon{\x}{\y}{\hexradius}
                
                \ifnum\ring=\gridradius
                    \ifnum\side=2
                        \drawbottomleftside{\x}{\y}{\hexradius}
                    \fi
                    \ifnum\side=3
                        \drawbottomleftside{\x}{\y}{\hexradius}
                    \fi
                    \ifnum\side=4
                        \drawbottomleftside{\x}{\y}{\hexradius}
                    \fi
                    
                    \ifnum\side=2
                        \ifnum\pos=\ring
                            \drawbottomhexagon{\x}{\y}{\hexradius}
                        \fi
                    \fi
                    \ifnum\side=3
                        \drawbottomhexagon{\x}{\y}{\hexradius}
                    \fi
                    \ifnum\side=4
                        \drawbottomhexagon{\x}{\y}{\hexradius}
                    \fi
                    
                    \ifnum\side=4
                        \drawbottomrightside{\x}{\y}{\hexradius}
                    \fi
                    \ifnum\side=5
                        \drawbottomrightside{\x}{\y}{\hexradius}
                    \fi
                \fi

                \ifnum\side=2
                    \ifnum\pos=0
                        \node[scale=3] at (\x, \y) {\number\numexpr\ring+1\relax};
                    \fi
                \fi
            }
        }
    }




    \pgfmathsetmacro\r{\hexradius}

    \pgfmathsetmacro\xo{\r*cos(60)}
    \pgfmathsetmacro\yo{\r*sin(60)}

    \pgfmathsetmacro\r{5*\hexradius}

    \draw[okabe3,line width=1] (0, 0) circle (\r);

    \pgfmathsetmacro\x{\r*cos(60)}
    \pgfmathsetmacro\y{\r*sin(60)}

    \draw[okabe4,line width=2,opacity=0.5] (0, 0) to node[pos=0.45,right,scale=4,xshift=-1.5pt,opacity=0.8] {$r$} (\x, \y);

    \fill[okabe7] (\x, \y) circle (0.15);
    \node[scale=3] at (\centerX, \centerY) {1};
\end{tikzpicture}}
    \caption{\label{fig:hex-lattice}
      A hexagonal lattice with $L = 4$ layers of rings%
      .
    }
    \vspace*{-\medskipamount}
\end{figure}
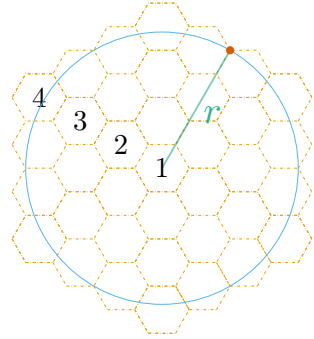

Increasing the number of rings increases the number of hexagons per lattice,
which in turn reduces each hexagon's size.
These smaller hexagons require smaller probes, reducing the number of responses
from the POI but increasing the (worst case) total number of probes.
Algorithm 1 can be thought of as one such algorithm where $L = 2$.
We describe our family of algorithms by the following routine:

\crefname{enumi}{step}{steps}
\Crefname{enumi}{Step}{Steps}

\begin{enumerate}
    \item \label{step:lattice} Cover the search area with an $L$-layer 
    lattice.

    \item \label{step:probe} Sequentially probe each hexagon in the lattice
        until receiving a positive response.

    \item \label{step:repeat} Repeat
    until
    reaching an area with
    radius 1.
\end{enumerate}

\begin{theorem} \label{thm:hexagonal-family}
    If a POI is only allowed to respond at most $1 \leq R_\text{max} \leq
    \lceil \log{n} \rceil$ times, then a hexagonal algorithm is able to find them
    using at most
    \begin{equation}
        P(n) \leq 6 R_\text{max} \binom{\lceil \frac{2n^{\frac{1}{R_\text{max}}}+2}{3}\rceil}{2} \text{ probes,}
    \end{equation}
    by using a hexagonal lattice with
    \begin{equation} \label{eq:hexagonal-rings}
        L \geq \lceil \frac{2n^{\frac{1}{R_\text{max}}}+2}{3} \rceil \text{ layers.}
    \end{equation}
\end{theorem}


\begin{fact} \label{fact:circumcircle}
    A hexagon with side length $s$ is covered by a circumscribed circle with
    radius $s$.
\end{fact}




\begin{lemma} \label{lem:hex-lattice-side-length-circle}
    A circle with radius $r$ can be covered by an $L$-layer lattice of hexagons
    with side length $s = \frac{2r}{3L-2}$.
\end{lemma}

\begin{proof}
See appendix.
\end{proof}


\begin{fact} \label{fact:centered-hexagonal}
    An $L$-layer
    lattice has $1 + 6 \binom{L}{2}$ hexagons, determined
    by the $L$-th centered hexagonal number.
\end{fact}


\begin{proof}(of \Cref{thm:hexagonal-family})
    During \cref{step:lattice}, we cover the radius $r$ search area with a
    lattice of hexagons of side length $s$, where $s = \frac{2r}{3L-2}$ from
    Lemma~\ref{lem:hex-lattice-side-length-circle}.
    %
    From Fact~\ref{fact:circumcircle}, we probe each of these hexagons with
    a circle of radius $s$, reducing the search area of the next recursive layer
    by a factor of $\frac{3L}{2} - 1$.
    As the
    lattice grows,
    we reduce the
    search area by a greater factor.
    In order to finally reach a circle of radius 1, we require
    $\lceil \log_{\frac{3L}{2} - 1}(n) \rceil$ recursive rounds, each
    requiring one
    response from the POI.
    Solving for $L$, we obtain \Cref{eq:hexagonal-rings}.
    Finally, since we can probe all the hexagons but one per round, we
    probe at most $6 \binom{L}{2}$ hexagons per round (see
    Fact~\ref{fact:centered-hexagonal}), and since there are $R_\text{max}$
    rounds, the total number of probes is bounded by $P(n) \leq 6 R_\text{max}
    \binom{L}{2}$.
\end{proof}

\begin{corollary} \label{cor:hexagonal-family}
    The total number of probes required to find a POI is bounded as follows:
    \begin{enumerate}
        \item \label{item:one-response} If $R_\text{max} = 1$, then $P(n) \leq \frac{4 n^2}{3} + 6 n + 6 =
            \mathcal{O}(n^2)$.

        \item \label{item:two-responses} If $R_\text{max} = 2$, then $P(n) \leq \frac{8 n}{3} + 12 \sqrt{n} + 12 =
        \mathcal{O}(n)$.

        \item \label{item:log-responses} If $R_\text{max} = \lceil \log{n} \rceil$, then $P(n) \leq 6 \lceil \log{n}
        \rceil$.\footnote{This follows from the fact that $n^{\frac{1}{\lceil
        \log{n} \rceil}} \leq 2$.}
    \end{enumerate}
\end{corollary}


\section{Finding All POIs}
Once one POI is found, we shut off the tracking device so that it stops
responding to $\Delta$'s probes, yet the question remains---how should
$\Delta$ search for the rest of the POIs?
Since $\Delta$ is stateless with respect to its current search area, no
knowledge is gained about other POIs from the search for the first POI.
Even if $\Delta$ was able to retain its search path so far, the result of
probes are binary; any previous positive probe result cannot be relied on.
And even if the probe were able to determine the exact \textit{quantity} of
POIs within the area, it is possible for them to be in different search areas early
on, resulting in the probe needing to perform its search almost entirely from
scratch.
This can trivially happen if the POIs are far from each other, but may
also happen if the POIs are close to each other but on opposite sides of a
search area boundary.
Does there exist a \emph{coordinated} search strategy that performs better than
the \textit{incremental} strategy of independently searching for each POI?




Let us assume we have a search algorithm $\mathcal{A}(n)$ that is able to find a
single POI in $P(n) \leq c \lceil \log{n} \rceil$ probes, traveling a distance
of $D(n) \leq d n$.
We use $\mathcal{A}$ as a subroutine of the following strategy:

\begin{enumerate}
	\item \label{step:memoryless-find-first} Find an arbitrary POI using $\mathcal{A}(n)$.

	\item \label{step:memoryless-shutdown} Shut off the tracking device of the found POI.

	\item \label{step:memoryless-probe} Without moving $\Delta$, re-probe the
	area at radius 2, 4, 8, etc., until a probe returns a positive result (i.e.,
	another POI is found).

	\item \label{step:memoryless-find-next} Invoke $\mathcal{A}$ using this new radius to find another POI.

	\item \label{step:memoryless-repeat} Repeat \cref{step:memoryless-shutdown,step:memoryless-probe,step:memoryless-find-next} until all POIs are found.
\end{enumerate}

Assuming there are $k$ POIs, we have:
\begin{theorem} \label{thm:memoryless-find-all}
	The total number of probes ($P_\text{tot}$) and the total distance
	traveled by $\Delta$ ($D_\text{tot}$) during the memoryless search algorithm
	for all $k$ POIs is at most:
	\begin{align*}
		P_\text{tot} &\leq c \lceil \log{n} \rceil + (c + 1) (k - 1) \lceil \log{\overline{e}} \rceil, \\
		D_\text{tot} &\leq d n + 2d E,
	\end{align*}
	where $E < \text{OPT} (\lceil \log{k} \rceil + 1)$, $\overline{e} =
	\frac{E}{k - 1}$, and $\text{OPT}$ is the optimal tour length for the
	traveling salesperson problem (TSP) on the $k$ POIs.
\end{theorem}
%
\begin{proof}
	See appendix.
\end{proof}




%
%
\renewcommand{\emph}[1]{\textit{#1}}

\small
\bibliographystyle{plainurl}
\bibliography{refs}

\clearpage
%
%
\renewcommand{\emph}[1]{\textbf{\textit{#1}}}

\appendix

\section{Additional Related Work}

As mentioned in the introduction, the Marco Polo problem falls into a rich area
of study known as \emph{localization algorithms}; see, e.g., the survey by Han,
Xu, Duong, Jiang, and Hara~\cite{han2013localization}.
Our approach differs from the approaches used in this prior work,
however, in that we are interested in strictly combinatorial strategies,
where all we learn is a single in-or-out result from each probe,
rather than, say, range and/or directional results,
such as in the work by Martinson and Dellaert~\cite{martinson}.


The Marco Polo problem is related to 
combinatorial group testing,
see, e.g., \cite{du1999combinatorial,eppstein2007improved,goodrich2008improved},
which was originally directed
at identifying WWII soldiers with syphilis~\cite{dorfman} 
and was recently applied to COVID-19 testing~\cite{covid}.
In this problem,
one is given a set of $n$ items, at most
$d$ of which are ``defective.''
Subsets of the items (such as blood samples) can be 
pooled and tested as a group, such that if one of the items 
in the pool is defective, then the test for the pool will be positive.
Tests can be organized to
efficiently identify the defective items based on the outcomes
of the tests.
The Marco Polo problem
differs from combinatorial group
testing, however, in that the search space for the Marco Polo problem
is a geometric region and tests must be connected
geometric shapes (i.e., radius-$d$ balls using the Euclidean metric), 
whereas the search space in
combinatorial group testing is defined by a discrete set of $n$ items
and tests can be arbitrary subsets of these items.

%

There is also some work on SAR 
algorithms that use call-and-response
protocols, such as the CenWits system by Huang, Amjad, and 
Mishra~\cite{cenwits},
which uses RF-based sensors
for the search and rescue of people, such as hikers, 
who are carrying mobile wireless communication devices
in wilderness areas. 

In the context of computational geometry, the Marco Polo problem is
somewhat related to the Freeze-Tag 
problem~\cite{arkin2006freeze,hammar2006online,arkin2003improved,bonichon2024euclidean,bonichon2024freeze,pedrosa2023freeze}, 
which involves ``waking up'' a collection for moving robots that are initially
at given points in the plane via a strategy motivated from the
children's game, ``Freeze Tag''~\cite{wiki:freeze}.

We stress that we are interested in
solutions to the Marco Polo problem that are adaptive,
where the $i$-th probe can depend on the results of the probes
that came earlier.
A non-adaptive solution to the Marco Polo problem would be related
to a constructive solution to a classic disk covering problem,
which asks for the
minimum number of disks of radius $\epsilon>0$ that can cover a 
region in the plane~\cite{cover}.

\section{Supplemental Pseudocode}

Pseudocode bounding the distance traveled by the search point, $\Delta$,
is provided in 
    Figure~\ref{fig:distance-bound-code}.

\begin{figure}[htb]
    \begin{algorithm}[H]
    \caption{Bounding the distance traveled by $\Delta$}
    \label{alg:distance-bound}
    \begin{algorithmic}[1]
    \Require A probe placement as a list of tuples $(x, y, \rho)$
    \Ensure Upper bound on the total distance traveled
    \State $d_k \gets 0$
    \State $p_{\text{curr}} \gets (0, 0)$ \Comment{Current position}
    \State $b_{\text{max}} \gets 0$ \Comment{Maximum bound}

    \For{each probe $(x, y, \rho_k)$ in placement}
        \State $(c_x, c_y) \gets p_{\text{curr}}$
        \State $d_k \gets d_k + \sqrt{(x - c_x)^2 + (y - c_y)^2}$

        \If{probe is last in placement}
			\LineComment{Skip round trip to final probe center}
			\vspace{\medskipamount}
            \State $(n_x, n_y) \gets \text{first probe in placement}$
			\State $d_1 \gets \sqrt{n_x^2 + n_y^2}$
            \State $d_k \gets d_k - 2 d_1 \rho_k$ \Comment{Subtract round trip}
        \EndIf

        \State $b_{\text{max}} \gets \max(b_{\text{max}}, \frac{d_k}{1 - \rho_k})$
        \State $p_{\text{curr}} \gets (x, y)$
    \EndFor

    \State \Return $b_{\text{max}}$
    \end{algorithmic}
    \end{algorithm}
    \caption[Algorithm for bounding the distance traveled by $\Delta$.]{Algorithm for bounding the distance traveled by $\Delta$.
        When the POIs are determined to be in the last probe area, $\Delta$
        can move directly to the first probe of the next search area, saving distance.
        If the first probe is at distance $d_1$ from the origin of the original
        search area, it is at distance $d_1 \rho_k$ from the origin of the
        new search area.
        The probe placement within the new search area can be rotated such that
        the first probe is as close to $\Delta$ as possible, saving an entire $d_1$-length
        round trip to the origin and back.\footnotemark
    }
    \label{fig:distance-bound-code}
\end{figure}

\footnotetext{This assumes there is not
        significant overlap between the last two probes, which holds for all our
        algorithms.}

\section{Omitted Proofs}
In this appendix, we provide proofs that were omitted in the body
of this paper.

\begin{lemma} 
(Same as Lemma~\ref{lem:hex-lattice-side-length-circle})
    A circle with radius $r$ can be covered by an $L$-layer lattice of hexagons
    with side length $s = \frac{2r}{3L-2}$.
\end{lemma}

\begin{proof}
    Let us consider the case of a hexagonal lattice with an even number of
    layers, as in \Cref{fig:hex-lattice}.
    The closest point from the center of the lattice to its boundary is obtained
    by moving diagonally along the center hexagons maximal diameter.
    Consider the distance, $r$, to that point.
    Every even-numbered layer of hexagons contributes one side length
    $s$, while every odd numbered layer contributes its maximal diameter $2s$,
    except for the first layer which contributes $s$.
    Thus, the total circumradius $r = \frac{3L}{2}s - s = (\frac{3L}{2} -
    1)s$.\footnote{This bound is tight for lattices with an even number of
    layers, and only improves for lattices with an odd number of layers.}
\end{proof}

\begin{theorem}[Same as \Cref{thm:memoryless-find-all}]
	The total number of probes ($P_\text{tot}$) and the total distance
	traveled by $\Delta$ ($D_\text{tot}$) during the memoryless search algorithm
	for all $k$ POIs is at most:
	\begin{align*}
		P_\text{tot} &\leq c \lceil \log{n} \rceil + (c + 1) (k - 1) \lceil \log{\overline{e}} \rceil, \\
		D_\text{tot} &\leq d n + 2d E,
	\end{align*}
	where $E < \text{OPT} (\lceil \log{k} \rceil + 1)$, $\overline{e} =
	\frac{E}{k - 1}$, and $\text{OPT}$ is the optimal tour length for the
	traveling salesperson problem (TSP) on the $k$ POIs.
\end{theorem}

\begin{proof}
    We re-iterate the steps of our stateless below:
    \begin{enumerate}
        \item Find an arbitrary POI using $\mathcal{A}(n)$.

        \item Shut off the tracking device of the found POI.

        \item Without moving $\Delta$, re-probe the
        area at radius 2, 4, 8, etc., until a probe returns a positive result (i.e.,
        another POI is found).

        \item Invoke $\mathcal{A}$ using this new radius to find another POI.

        \item Repeat \cref{step:memoryless-shutdown,step:memoryless-probe,step:memoryless-find-next} until all POIs are found.
    \end{enumerate}

	Other than \cref{step:memoryless-find-first}, the performance of the search
	strategy depends by the relative positions of the POIs.
	Let the first POI found be POI 0, and the second POI found in
	\cref{step:memoryless-find-next} be POI 1.
	Let the distance between the two POIs be denoted as $e_1$, and the
	distance between POI $i - 1$ and POI $i$ be denoted as $e_i$, for $i \in
	\{1, 2, \ldots, k\}$.
	The total number of probes in \cref{step:memoryless-probe} required to find a
	large enough search area containing POI $i$ is $\lceil \log{e_i} \rceil$,
	where the size of the search area is $2^{\lceil \log{e_i} \rceil}$.
	Finding the POI within this search area (\cref{step:memoryless-find-next})
	will take $P(2^{\lceil \log{e_i} \rceil}) = c \lceil \log{2^{\lceil
	\log{e_i}\rceil}} \rceil = c \lceil \log{e_i}\rceil$ probes.
	Adding on the initial probe to find POI 0, the total number of probes
	required to find all POIs is at most:
	\begin{equation*}
		P_\text{tot} \leq c \lceil \log{n} \rceil + \sum_{i=1}^{k - 1} (c + 1) \lceil \log{e_i} \rceil.
	\end{equation*}

	Let $E = \sum_{i=1}^{k - 1} e_i$, and $\overline{e} = \frac{E}{k - 1}$.
	At each step of the algorithm, we find a POI that is within a factor of
	two of the closest still-undiscovered POI to the last found POI.
	This algorithm will therefore perform at worst a factor of two approximation
	of the nearest neighbor tour for the traveling salesperson problem
	(TSP), e.g., see \cite{HOUGARDY2015101,BRECKLINGHAUS2015259}.
	In Euclidean space, the (greedy) nearest neighbor tour of $n$ salespeople
	performs at worst a factor of $\frac{1}{2}(\lceil \log{n} \rceil + 1)$ of
	the optimal TSP tour length (OPT), see \cite{nnupper}.
	Applying to our case, $E < \text{OPT} (\lceil \log{k} \rceil + 1)$.
	Since $\log{e}$ is a concave function, it
    follows from Jensen's inequality
	that $\sum_{i=1}^{k - 1} \lceil \log{e_i} \rceil  \leq (k - 1) \lceil
	\log{\overline{e}} \rceil$,
	where $\overline{e}$ is the average value of
	$e$,
	and we obtain our desired result.

	Regarding the total distance traveled, note that $\Delta$ only moves during
	\cref{step:memoryless-find-first,step:memoryless-find-next}.
	For the first POI, $\Delta$ travels $D(n) \leq d n$.
	Using similar reasoning as above, for the $i$-th POI, $\Delta$ travels
	$D(2^{\lceil \log{e_i} \rceil}) \leq d (2 e_i)$.
	Overall, the total distance traveled is at most: $d n + 2d \sum_{i=1}^{k -
	1} e_i = d n + 2d E$, and our result follows.
\end{proof}

\section{Algorithm Assumptions} \label{sec:assumptions}

In this section we justify some assumptions regarding the possible probes and
locations of the POIs that we make in the main body of the paper.
Namely, we assume that:
\begin{itemize}
\item
The maximum allowed probe distance, $d$, is $n$.
\item
There may be multiple POIs,
either within the initial search region
or slightly outside of it.
\end{itemize}
More specifically, we show how, by relaxing these assumptions, a relatively
simple solution is able to solve the problem using an optimal number of probes.

\subsection{Simple Case: Exactly One POI at Full Distance}
\label{sec:bad}
Let us begin
with a simple (but admittedly unrealistic) scenario,
where there is exactly one POI within distance~$n$ of the origin
and we can perform arbitrarily large probes.
In this case, where
we allow unbounded probe distances,
we can find the sole POI using two
binary searches, one in each dimension.
We place $\Delta$ at a very large distance from the origin, such that the
intersection for the probe with the original search region is nearly a
straight line,
and perform one binary search using $\lceil \log{n} \rceil + 1$ probes to find a
region of width at most 1 in the first dimension.
We then perform a similar binary search placing $\Delta$ in an orthogonal
direction to find the POI in the second dimension.
We are left with a region of side length 1, which is fully contained
within a single, final, radius-1 probe.
This solution thus
uses $2 \lceil \log{n} \rceil + \mathcal{O}(1)$ probes,
but it relies on $\Delta$ traveling very far away and at high
altitude and the tracking device to have very strong signal strength,
which are not reasonable assumptions.
This can be improved somewhat
by restricting the probe distance to $n$,
as we explore in the next section.

\subsection{A Slight Improvement to a Bad Algorithm}
%
While the previous algorithm allows for unbounded probe distances,
we now show how to achieve a similar result while restricting the probe size to
be at most $n$.
In this case, we
reduce the horizontal dimension this time first to a
1-width arc, and then cut the length of the arc in half for each subsequent
probe.
Note that the arc may be up to $\pi n$ units long, which is greater than
the $2n$ original region diameter, resulting in up to one more probe.
This probe can be removed by reducing the
remaining \textit{area} by two per probe in the
first dimension rather than the arc width.
See \Cref{fig:triv-one}.

\begin{figure}
	\centering
	\resizebox*{0.45\linewidth}{!}{\begin{tikzpicture}
	\def\n{8}

	\tikzset{
		clip even odd rule/.code={\pgfseteorule}, 
		invclip/.style={
			clip,insert path=
				[clip even odd rule]{
					[reset cm](-\maxdimen,-\maxdimen)rectangle(\maxdimen,\maxdimen)
				}
		}
	}

	\begin{scope}
		\clip (-\n-1,-\n-1) rectangle (\n+1,\n+1);
		
		\begin{scope}
			\clip[invclip] (4,0) circle (\n);
			\clip (0,0) circle (\n);
			\clip (3,0) circle (\n);

			\fill[pattern={Lines[angle=-45,distance=6]},pattern color=okabe3!50] (-3.5,8) circle (\n);
		\end{scope}

		\begin{scope}
			\clip[invclip] (4,0) circle (\n);
			\clip[invclip] (-7.5,1) circle (\n);
			\clip[invclip] (-3.5,8) circle (\n);
			\clip (0,0) circle (\n);

			\fill[pattern={Lines[angle=70,distance=6]},pattern color=okabe7!40] (3,0) circle (\n);
		\end{scope}

		\draw[line width=2,loosely dashed] (0,0) circle (\n);
		\draw[color=okabe2,dashed] (4,0) circle (\n);
		\draw[color=okabe2,dashed] (3,0) circle (\n);

		\draw[color=okabe3] (-3.5,8) circle (\n);
		\draw[color=okabe7] (-7.5,1) circle (\n);

		\node[scale=4] at (-4.3,-1) {\ding{55}};
	\end{scope}

		






\end{tikzpicture}}
	\caption{\label{fig:triv-one}
		A search for one POI (\ding{55}), first reducing the horizontal dimension to a
		width of 1 (in the dashed orange probes), then searching along the remaining arc
		(with the solid blue and orange probes).
	}
\end{figure}
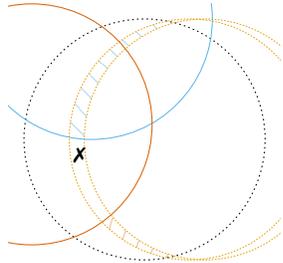

Note that this algorithm may require $\Delta$ to travel outside of our original
search region.
This can be avoided with a final algorithm.
We can place $\Delta$ at the center of the original search region and perform
a binary search by changing the radius of the probes, resulting in at most a
width-1 shell of the original search region, requiring at most $\lceil \log{n} \rceil$
probes.\footnote{Note that reducing the remaining \textit{area} by a factor of
two at this stage would reduce a constant number of probes.}
This shell of outer-radius $r$ is $2 \pi r$ units long, where $r$ can be up to $n$.
We can reduce this shell using a similar binary search, by moving $\Delta$
along the outer edge of the shell, performing probes with radius $r$.
See \Cref{fig:triv-one-inside} for an example of this algorithm.
Note that the first of these probes will only reduce one third of the shell,
rather than a factor of 2, resulting in up to one more probe overall.
We can then perform probes to reduce this shell to a width-1 square, which could
be performed in $\lceil \log{n} \rceil + 2$ probes.

\begin{figure}[hbt!]
\centering
\resizebox*{.5\linewidth}{!}{\begin{tikzpicture}
	\def\n{8}

	\tikzset{
		clip even odd rule/.code={\pgfseteorule}, 
		invclip/.style={
			clip,insert path=
				[clip even odd rule]{
					[reset cm](-\maxdimen,-\maxdimen)rectangle(\maxdimen,\maxdimen)
				}
		}
	}

	\begin{scope}
		\clip (-\n-1,-\n-1) rectangle (\n+1,\n+1);
		




		\draw[line width=1] (0,0) circle (\n);
		\draw[color=okabe2,dashed] (0,0) circle (\n);
		\draw[color=okabe2,dashed] (0,0) circle (\n-1);

		\draw[color=okabe3] ({8*cos(0)},{8*sin(0)}) circle (\n);
		\draw[color=okabe3] ({8*cos(120)},{8*sin(120)}) circle (\n);
		\draw[color=okabe3] ({8*cos(240)},{8*sin(240)}) circle (\n);
	\end{scope}

		






\end{tikzpicture}}
\caption{\label{fig:triv-one-inside}
	A depiction of the final ``simple'' search strategy where it is known that
	there is only one POI within the search region.
	In this final strategy, not only is the probe distance limited to $n$, but
	$\Delta$ is also restricted to the original search region.
}
\end{figure}
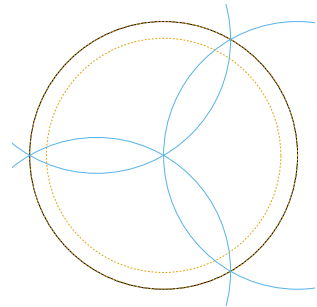

\section{Darting Non-Monotonic Algorithms}
In this section we discuss the darting algorithms, which are non-monotonic
and place many more probes at each recursive level by
first placing several carefully-placed
probes and then greedily adding more to fill in smaller and smaller gaps,
albeit at the expense of darting from side to side in the search space
to do so.
That is, the approach for our
darting algorithms
is as follows: after an initial placement of several probes,
we determine if there are any uncovered internal areas.
If so, we find the largest such area and place a
probe such that it intersects two
points on the area's convex hull.
If there are multiple possible placements, we repeatedly
choose the one that reduces the remaining area the most in a greedy fashion.
See \Cref{fig:alg-4-fail}.
Subsequent probes are placed in the same manner, until either the entire search
area is covered, or the probes become sufficiently small as to be unable to
cover the remaining area.

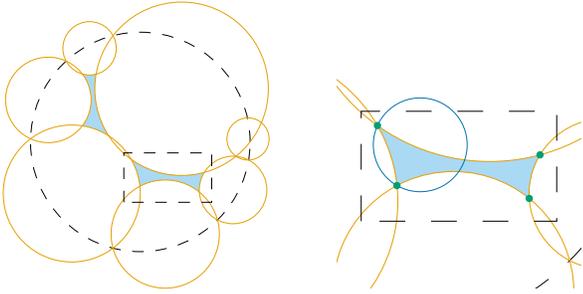
\begin{figure}[htb!]
	\centering
	\begin{subfigure}{0.46\linewidth}
		\centering
		\resizebox*{\linewidth}{!}{

\begin{tikzpicture}
	\def\circles{
        {0.3758999999999999/0.48435440536863084/0.79},
        {-0.6242412976006705/-0.469917005831701/0.6241000000000001},
        {0.22705745900761443/-0.8398556154411339/0.49303900000000006},
        {-0.837645132262955/0.382935701398447/0.3895008100000001},
        {0.8446632757242072/-0.4380196226387391/0.30770563990000005},
        {-0.4619510485116507/0.8529418020869375/0.24308745552100006},
        {0.9810128381197176/0.02710718372367052/0.19203908986159007}%
    }

	\tikzset{
		clip even odd rule/.code={\pgfseteorule}, 
		invclip/.style={
			clip,insert path=
				[clip even odd rule]{
					[reset cm](-\maxdimen,-\maxdimen)rectangle(\maxdimen,\maxdimen)
				}
		}
	}

	\begin{scope}
		\clip (0,0) circle (1);

		\foreach \x/\y/\r in \circles {
			\clip[invclip] (\x,\y) circle (\r);
		}

		\fill[okabe3!50] (0,0) circle (1);
	\end{scope}

	\draw[line width=0.1,dashed] (0,0) circle (1);
    \foreach \x/\y/\r in \circles {
      \draw[line width=0.1,okabe2] (\x,\y) circle (\r);
    }

	\draw[line width=0.1,dashed] (-0.15, -0.55) rectangle (0.65, -0.1);
\end{tikzpicture}}

	\end{subfigure}
	\hfill
	\begin{subfigure}{0.45\linewidth}
		\centering
		\resizebox*{\linewidth}{!}{\begin{tikzpicture}
	\def\circles{
        {0.3758999999999999/0.48435440536863084/0.79},
        {-0.6242412976006705/-0.469917005831701/0.6241000000000001},
        {0.22705745900761443/-0.8398556154411339/0.49303900000000006},
        {-0.837645132262955/0.382935701398447/0.3895008100000001},
        {0.8446632757242072/-0.4380196226387391/0.30770563990000005},
        {-0.4619510485116507/0.8529418020869375/0.24308745552100006},
        {0.9810128381197176/0.02710718372367052/0.19203908986159007}%
    }

	\tikzset{
		clip even odd rule/.code={\pgfseteorule}, 
		invclip/.style={
			clip,insert path=
				[clip even odd rule]{
					[reset cm](-\maxdimen,-\maxdimen)rectangle(\maxdimen,\maxdimen)
				}
		}
	}

	\begin{scope}
		\clip (-0.25, -0.825) rectangle (0.75, 0.175);

		\begin{scope}
			\clip (0,0) circle (1);

			\foreach \x/\y/\r in \circles {
				\clip[invclip] (\x,\y) circle (\r);
			}

			\fill[okabe3!50] (0,0) circle (1);
		\end{scope}

		\draw[line width=0.1,dashed] (0,0) circle (1);
		\foreach \x/\y/\r in \circles {
			\draw[line width=0.1,okabe2] (\x,\y) circle (\r);
		}

		\draw[line width=0.1,okabe6] (0.09157683073435136, -0.23732648590651592) circle (0.19203908986159007);

		\node (A) at (0.5375346952702428,-0.4568526636521755) {};
		\node (B) at (-0.0036205457305954728, -0.4041091934905429) {};
		\node (C) at (-0.08347338962420232, -0.15835540013906763) {};
		\node (D) at (0.5815987142614405, -0.27839584339364054) {};

		\fill[okabe4] (A) circle (0.015);
		\fill[okabe4] (B) circle (0.015);
		\fill[okabe4] (C) circle (0.015);
		\fill[okabe4] (D) circle (0.015);


		\draw[line width=0.1,dashed] (-0.15, -0.55) rectangle (0.65, -0.1);
	\end{scope}
\end{tikzpicture}}
	\end{subfigure}
	\caption{\label{fig:alg-4-fail}
		A method to greedily add an arbitrary amount of probes to any
		initial placement of probes.
		On the left, Algorithm 4 given an insufficient $\rho_1$ value fails to
		cover the search area, leaving some uncovered internal area (in blue).
		On the right, the largest uncovered internal area's convex hull
		(in green) is identified, and a probe is placed such that it intersects
		two points on the hull.
	}
\end{figure}

We note that
the final probe of Algorithm~4 is often the least efficient; see
\Cref{fig:alg-4-fail} (left) and \Cref{fig:alg-3-4}.
If we run Algorithm~4, remove the final probe, and then apply the greedy method
described above to place the remaining probes, we obtain Algorithm~7, which is
able to significantly reduce the number of probes required to find a POI to
$P(n)  < 2.93 \lceil \log{n} \rceil$.

Exploiting this approach further, however, strains our ability
to reason about regions that are uncovered after performing many probes;
hence, for an even further improved algorithm,
Algorithm~8, its probe sequence for each recursive level
is determined using computer-assisted proof.
More specifically, the placement of the first six probes is determined by a
differential evolution algorithm, which is a type of genetic algorithm that
optimizes a function by iteratively improving a population of candidate
solutions, see~\cite{Storn1997}.
Next, we fill in the gaps according to the
aforementioned
greedy method.
The resulting
Algorithm~8
is able to achieve
our best probe results,
with $P(n)  < 2.53 \lceil \log{n} \rceil$, albeit with
a very large value for $D(n)$.
See Figure~\ref{fig:alg-7-8}.

\section{Experiments}

We implemented our 8 algorithms and tested them.
The data were obtained by placing a POI at a random location, determined from a
uniformly random angle and a random distance from the center of the search area.
When the POI was in the last probe's search area, the last probe was not
executed, and $\Delta$ proceeded directly to the first probe of the next search
area.
Each algorithm was executed 40 million times, where $n = 2^{20}$, and normalized
by dividing by either $\lceil \log{n} \rceil$ or by $n$.
Both Algorithms~1 and~2, despite having a poor worst-case probe complexity,
perform well in practice, finding the POI using fewer probes on average
than Algorithms~3--6.
Our computer-assisted algorithms (7 and~8), however, outperform them.
Interestingly,
the progressively shrinking algorithms all have a non-zero
variability in their probe counts.
See \Cref{fig:P}.
While Algorithms~1 and~2 perform well in terms of average distance traveled,
they are outperformed by Algorithms~5 and~6, which each have the best
distance-traveled guarantees.
See \Cref{fig:D}.
Overall, the best methods are Algorithm~8 if number of probes
is a priority, Algorithm~6
if distance traveled is a priority, and Algorithm~2 for a good balance.


\begin{figure}[hbt!]
	\centering
	\includegraphics[width=0.62\linewidth]{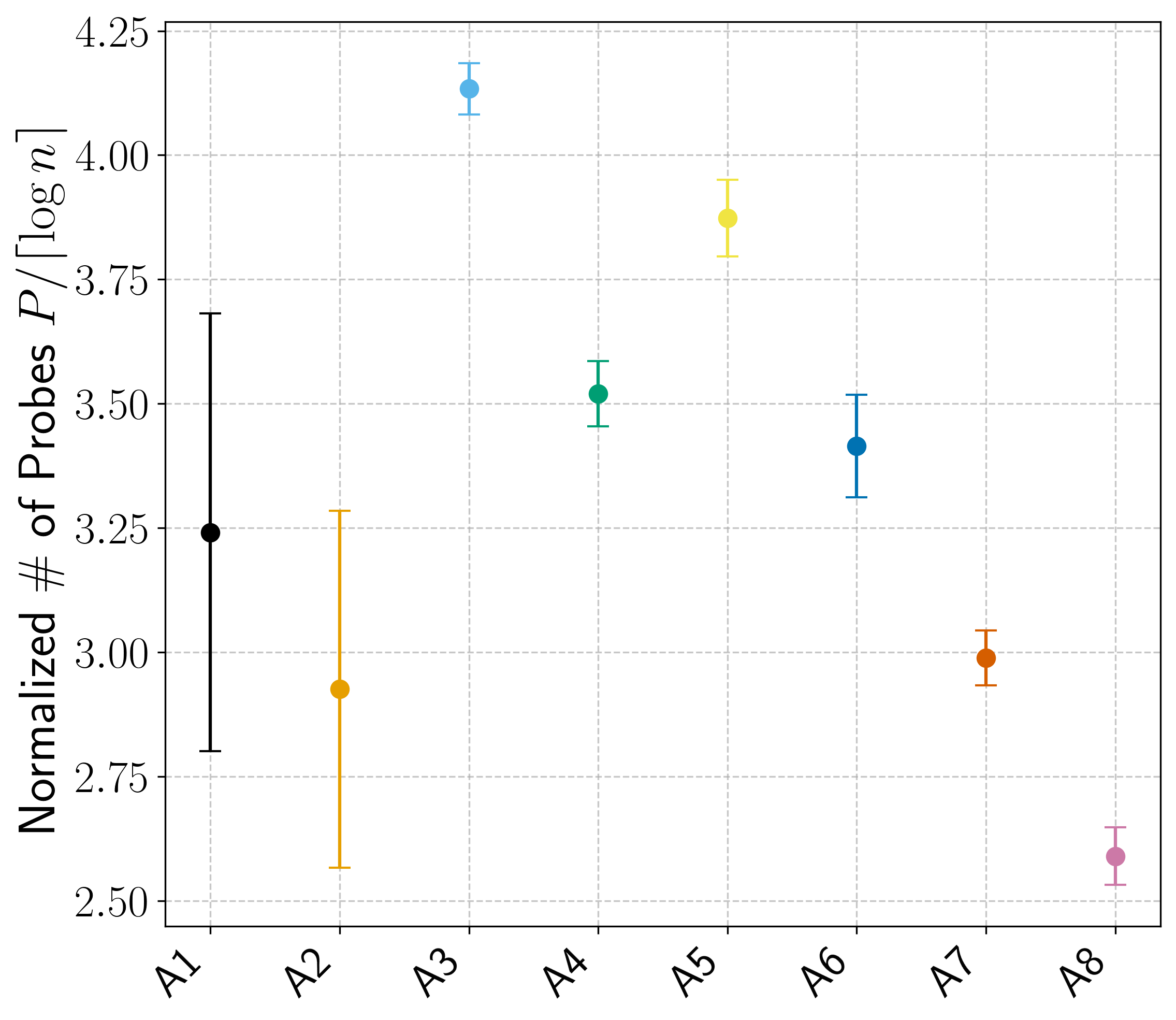}
	\caption{
		Simulation results for $P / \lceil \log{n} \rceil$%
		.
		Error bars represent one standard deviation from the mean.
	}
	\label{fig:P}
\end{figure}

\begin{figure}[htb!]
	\centering
	\includegraphics[width=0.62\linewidth]{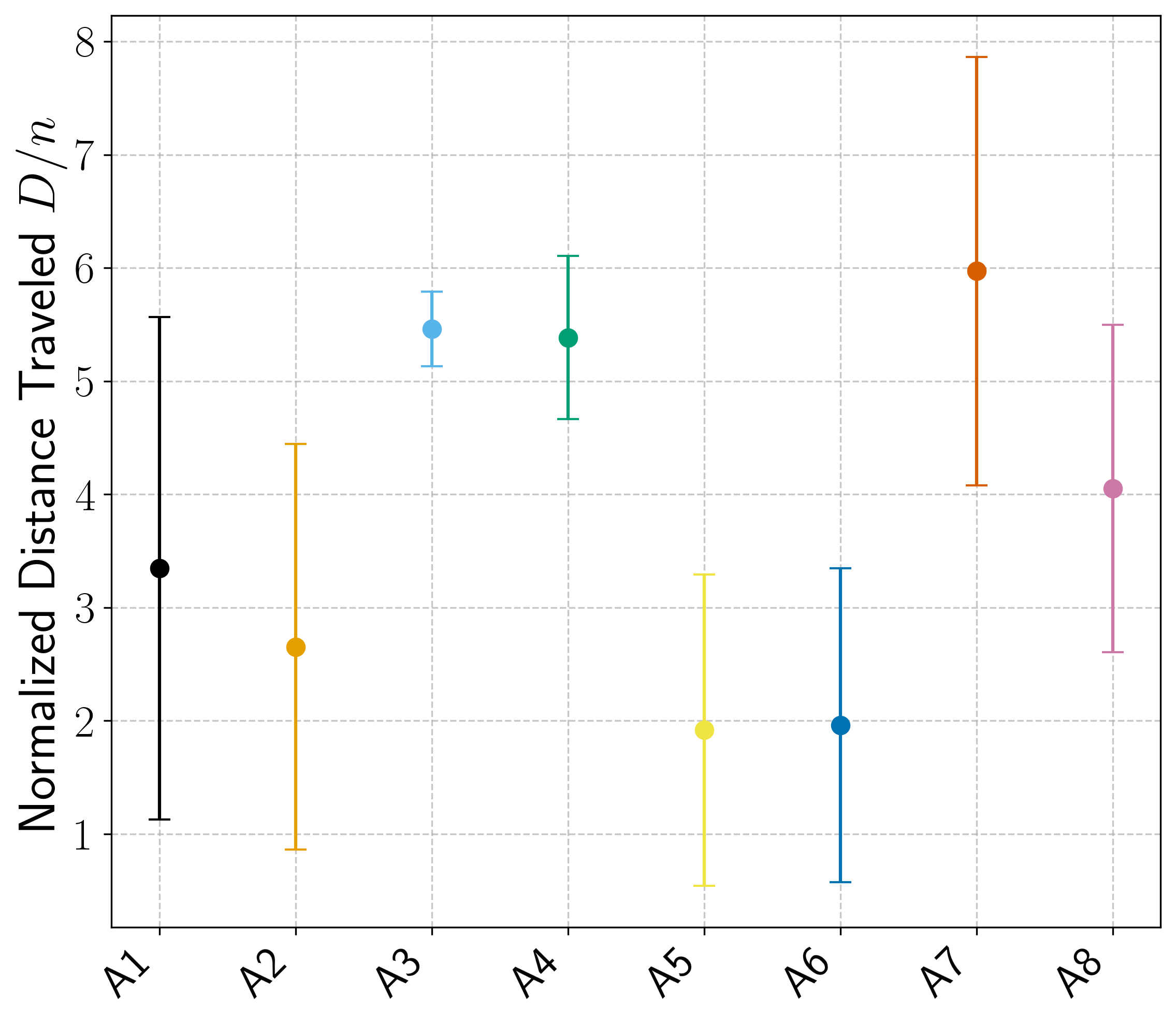}
	\caption{
		Simulation results for $D / n$%
		,
		where $n = 2^{20}$.
	}
	\label{fig:D}
\end{figure}

\begin{table*}[tb!]
	\centering
	\begin{tabular}{|l|l|rrll|rrll|rll|}
	\hline
	& & \multicolumn{4}{c|}{Probes ($P/\lceil \log{n} \rceil$)} & \multicolumn{4}{c|}{Total Distance ($D/n$)} & \multicolumn{3}{c|}{Responses ($R/\lceil \log{n} \rceil$)} \\
	Category & Alg. \# & Min & Avg & Max & Bound & Min & Avg & Max & Bound & Avg & Max & Bound \\
	\hline
	\multirow{2}{*}{Hexagonal} & Alg. 1 & \textbf{1.00} & 3.24 & 5.70 & 6.00 & \textbf{0.00} & 3.35 & 10.39 & 10.39 & \textbf{0.89} & \textbf{1.00} & \textbf{1.00} \\
	 & Alg. 2 & \textbf{1.00} & 2.93 & 4.80 & 5.00 & \textbf{0.00} & 2.65 & 8.81 & 8.81 & 1.11 & 1.45 & 2.00 \\
	\hline
	\multirow{2}{*}{Chord-Based} & Alg. 3 & 3.85 & 4.13 & 4.25 & 4.08 & 4.69 & 5.46 & 6.56 & 6.95 & 1.99 & 2.40 & 4.08 \\
	 & Alg. 4 & 3.10 & 3.52 & 3.70 & 3.54 & 4.30 & 5.38 & 9.00 & 9.31 & 1.94 & 2.50 & 3.54 \\
	\hline
	\multirow{2}{*}{Monotonic} & Alg. 5 & 3.55 & 3.87 & 4.15 & 3.83 & \textbf{0.00} & \textbf{1.92} & 6.72 & 6.72 & 2.49 & 3.85 & 3.83 \\
	 & Alg. 6 & 3.25 & 3.41 & 3.85 & 3.34 & \textbf{0.00} & 1.96 & \textbf{6.01} & \textbf{6.02} & 1.96 & 3.35 & 3.34 \\
	\hline
	\multirow{2}{*}{Darting} & Alg. 7 & 2.90 & 2.99 & 3.65 & 2.93 & 3.86 & 5.97 & 25.74 & 25.80 & 1.39 & 2.15 & 2.93 \\
	 & Alg. 8 & 2.55 & \textbf{2.59} & \textbf{3.20} & \textbf{2.53} & 2.44 & 4.05 & 42.58 & 45.40 & 1.31 & 1.85 & 2.53 \\
	\hline
	\end{tabular}
	\caption{A numerical comparison of simulation results for our 8 algorithms on three normalized performance metrics, namely the number of probes made ($P$), the total distance traveled by $\Delta$ ($D$), and the number of POI responses ($R$). The best values are highlighted in bold. The category names used are crude abbreviations; see the main paper for their proper names.}
	\label{tab:algorithm_metrics}
\end{table*}

\subsection{Number of Probes Made}
We make two interesting observations regarding the number of probes made,
$P$, for our progressive probe algorithms.

\begin{observation}
	The progressive probe algorithms, Algorithms 3--8, exhibit a non-zero variance
	in $P$ experimentally.
\end{observation}

And perhaps more surprisingly, in \Cref{tab:algorithm_metrics} we observe:

\begin{observation}
	Progressive probe algorithms appear to perform more probes than their
	theoretical upper bounds.
\end{observation}

Both observations are primarily explained by the last recursive layer in each of
our algorithms.
Our theoretical analysis considered the number of probes needed to reduce the
search area to radius 1.
However, in practice, the probe may reduce the search area to be significantly
smaller.
For example, let's consider running Algorithm~8, with 33 probes per layer, on a
search area with radius 1 + $\epsilon$.
Algorithm~8's smallest probe has a proportionality factor $\rho_\text{min} =
0.00012$.
In other words, if the POI is unlucky enough to be in the last probe on this
last recursive layer, the search area radius will be reduced to $\approx
0.00012$, which is significantly smaller than 1, and this excessive precision
costs us extra probes.
In this worst case, where the POI is in the last probe of the last recursive
layer, we would be performing $|\text{probes}| - 1$ probes, when only 1 probe
was required.
A more precise upper bound of the number of probes $P^*(n)$, therefore, is
$P^*(n) = P(n) + |\text{probes}| - 2$.
However, for sufficiently large $n$, the coefficient in front of the $\lceil
\log{n} \rceil$ term in $P(n)$ will dominate this constant number of extra
probes, so our theoretical upper bounds will hold.

Another factor that contributes to the probe variability, albeit in the other
direction, is the fact that the last probe is omitted in each layer, while our
for the progressively sized probes assumes that every probe
is executed, and in fact supports an infinite number of probes per layer.

\subsection{Total Distance Traveled by the Search Point}
As expected, both our higher-count monotonic-path (HM) algorithms, Algorithms~5
and~6, minimize the total distance traveled by $\Delta$.
Algorithms~1--3 also have monotonic counter-clockwise paths, but their less
efficient probe sequences result in worse average distance traveled.
See \Cref{tab:algorithm_metrics}.
The hexagonal and HM algorithms each start with a
central probe,\footnote{Recall that the hexagonal algorithms can be modified to
start with the first probe instead of ending with it.} and consequently have a
minimum of 0 distance traveled.
While all algorithms experienced an instance where $\Delta$ traveled nearly as
much as their worst-case bound, reassuringly, most algorithms performed
significantly better on average, with the hexagonal and HM algorithms performing
\textasciitilde 3 times better, and the darting non-monotonic-path algorithms,
Algorithms~7 and~8, performing \textasciitilde 4 and \textasciitilde 11 times
better, respectively.
They are still likely not best suited for time-critical rescue
operations since they travel
2-3 times more than the HM algorithms
on average.

\subsection{Number of POI Responses}
Finally, we compare the number of POI responses, $R$, for our algorithms, to
see which are best suited for the case where the POI has battery constraints.
We find that, Algorithm~1 performs by far the best, with not only the fewest
responses across the board, with its maximum number of responses being lower
than the average of any other algorithm, but also with the smallest standard
deviation.
See \Cref{fig:R}.
This result is not surprising, because in
Corollary~\ref{cor:hexagonal-family}~(\ref{item:log-responses}), we learn that
Algorithm~1 is part of a larger family of response-efficient algorithms.
The next best algorithm is Algorithm~2, also partly based on a hexagonal grid,
and then the darting algorithms, Algorithms~7 and~8, which perform only a small
number of probes in general.

\begin{figure}[htb!]
	\centering
	\includegraphics[width=0.62\linewidth]{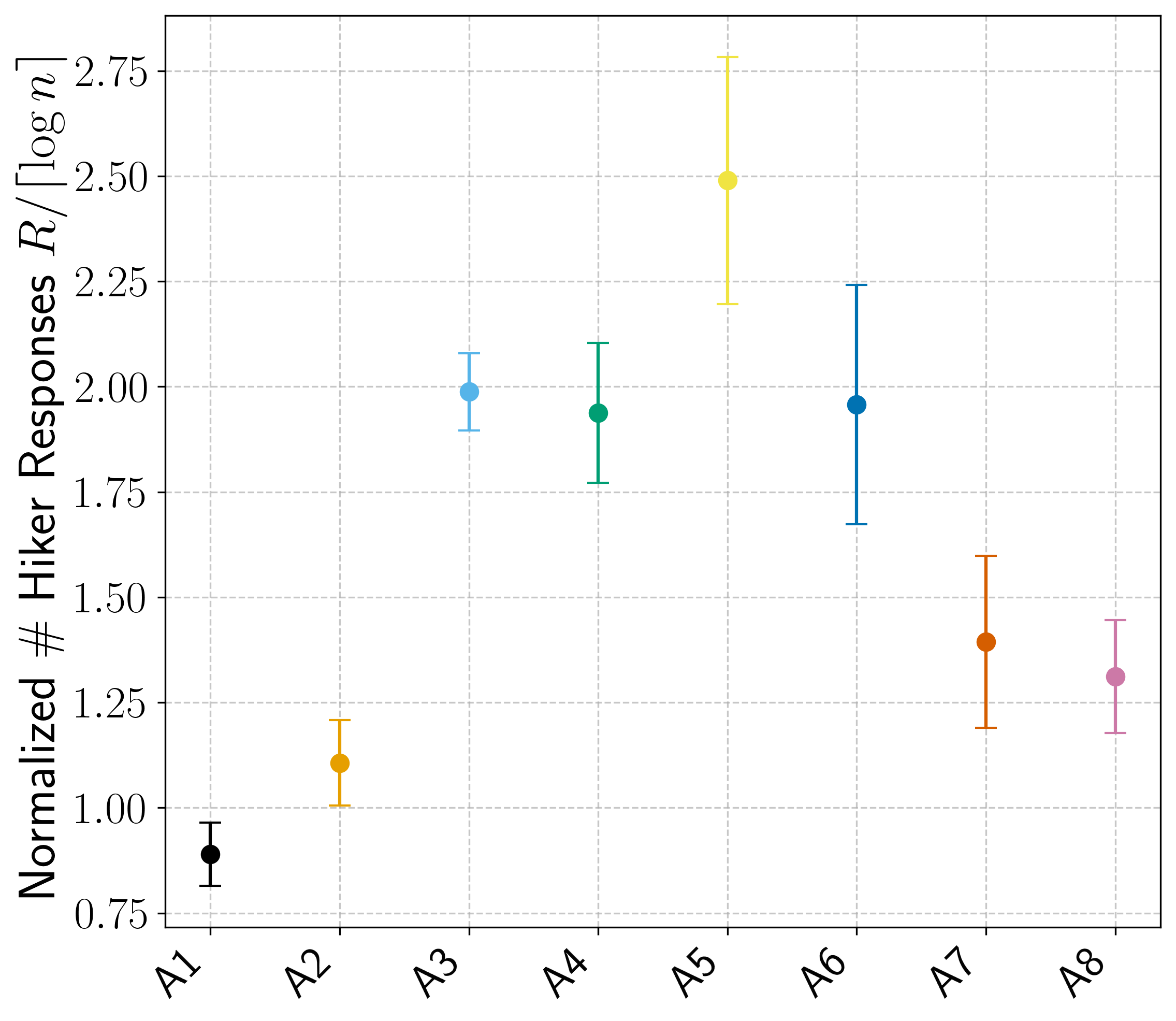}
	\caption{
		Simulation results for $R / \lceil \log{n} \rceil$%
		.
	}
	\label{fig:R}
\end{figure}

\end{document}